\documentclass[sn-mathphys-num]{sn-jnl}%
\usepackage{graphicx}
\usepackage{multirow}
\usepackage{amsmath, amssymb, amsfonts}
\usepackage{amsthm}
\usepackage{mathrsfs}
\usepackage[title]{appendix}
\usepackage{xcolor}
\usepackage{textcomp}
\usepackage{manyfoot}
\usepackage{booktabs}
\usepackage{algorithm}
\usepackage{algorithmicx}
\usepackage{algpseudocode}
\usepackage{listings}
\usepackage{longtable}
\usepackage{natbib}
\usepackage{hyperref} 
\usepackage{url}

\newtheorem{theorem}{Theorem}

\newtheorem{lemma}{Lemma}

\begin{document}

\title[Article Title]{A Method to Generate Multi-interval Pairwise Compatibility Graphs}

\author[1]{\fnm{Seemab Hayat}}\email{shayat@math.qau.edu.pk}

\author*[1,2]{\fnm{Naveed Ahmed Azam}\email{azam@amp.i.kyoto-u.ac.jp}}

\affil[1]{\orgdiv{Department of Mathematics}, \orgname{Quaid-i-Azam University}, \city{Islamabad}, \country{Pakistan}}

\affil[2]{\orgdiv{Department of Applied Mathematics and Physics}, \orgname{Kyoto University}, \city{Kyoto}, \country{Japan}}

\abstract{Reconstruction of evolutionary relationships between species is an important topic in the field of computational biology. 
Pairwise compatibility graphs (PCGs) are used to model such relationships. 
A graph is a PCG if its edges can be represented by the distance between the leaves of an edge-weighted tree within a fixed interval.
If the number of intervals is more than one, then the graph with such a tree representation is called a multi-interval PCG. 
The aim of this paper is to generate all multi-interval PCGs with a given number of vertices. 
For this purpose, we propose a method to generate almost all multi-interval PCGs corresponding to a given tree by randomly assigning edge weights and 
selecting typical intervals.  
To reduce the exponential tree search space, we theoretically prove that for each multi-interval PCG there exists a tree whose internal vertices have degree exactly three, and developed an algorithm to enumerate such trees. 
The proposed method is applied to enumerate all two-interval PCGs with up to ten vertices. 
Our computational results establish that all graphs with up to ten vertices are 2-IPCGs, making significant progress towards the open problem of determining whether a non-2-IPCG exists with fewer than 135 vertices.}

\maketitle

\section{Introduction}
The study of phylogeny is of significant importance in the field of computational biology~\cite{bayzid2010pairwise}. The term phylogeny describes the evolutionary relationships between species, and the process of phylogenetic reconstruction focuses on determining these relationships for groups of organisms. These relationships are represented using a tree structure known as a phylogenetic tree, which serves as a visual representation of the evolutionary history of the species. In a phylogenetic tree, there are two main components: leaves and branches. The leaves are located at the tip of the tree and symbolize the organisms that are being analyzed in terms of their evolutionary history. The evolution of the example DNA sequences is visualized as a phylogenetic tree in Fig.~\ref{fig:phylogeny_tree}.
\begin{figure}[h!]
        \centering
         \includegraphics[width=0.4\textwidth]{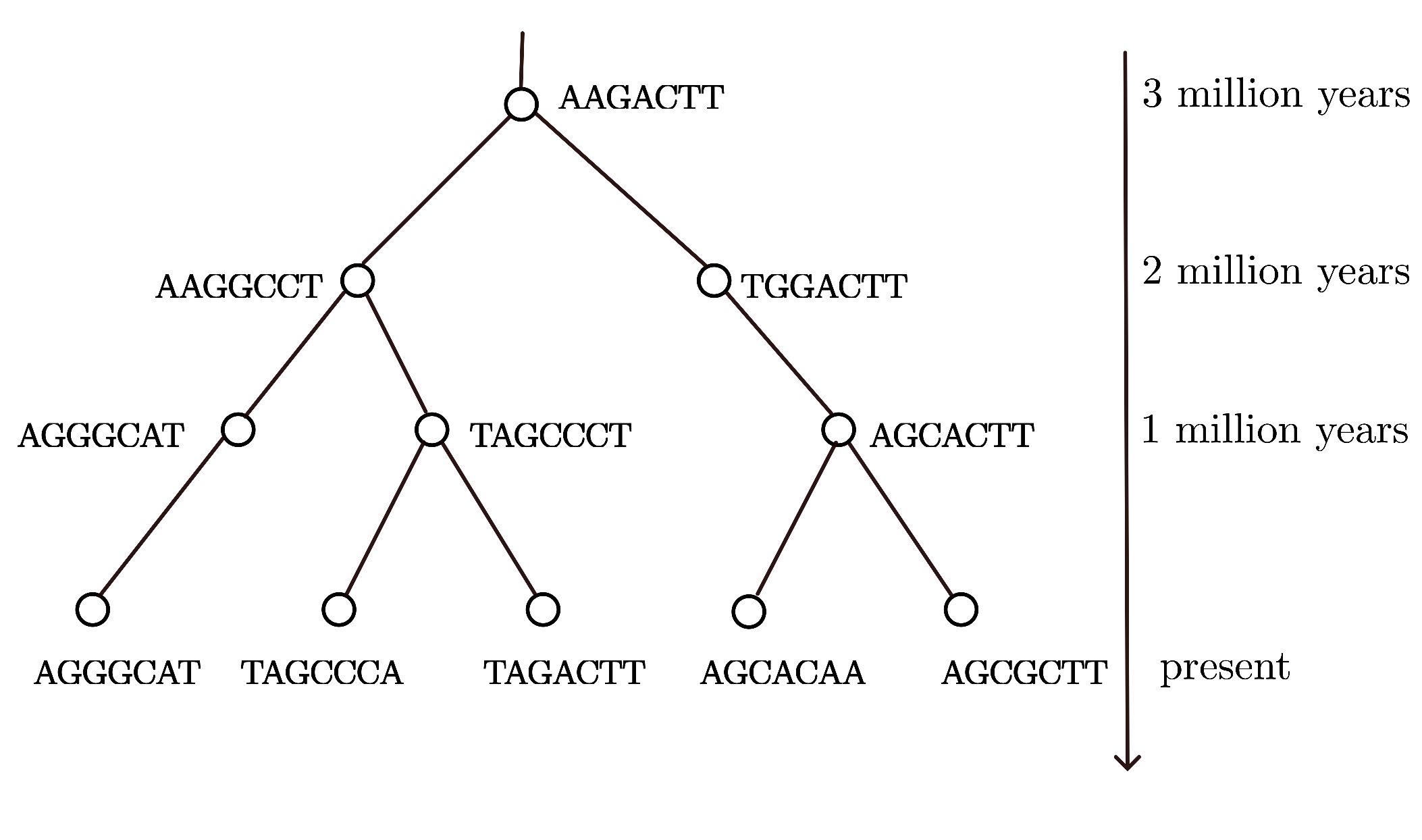}
        \caption{Phylogenetic tree of evolution of the DNA sequences~\cite{tandy2009}}
        \label{fig:phylogeny_tree}
    \end{figure}
      
Kearney et al.~\cite{kearney2003efficient} introduced the concept of {\em pairwise compatibility graphs (PCGs)} and demonstrated their application in modeling evolutionary relationships among groups of organisms. 
A graph with $n$ vertices is said to be a PCG if there exist an edge-weighted tree with $n$ leaves, an interval of non-negative real numbers, and a correspondence between the leaves of the tree and the vertices of the graph, such that two vertices are adjacent in the graph if and only if the distance (the sum of the weights of all edges on the path) between the corresponding leaves lies in the interval. 
Fig.~\ref{fig:PCG}(a) illustrates a PCG due to the edge-weighted tree given in Fig.~\ref{fig:PCG}(b), interval $[5, 10]$ and the identity correspondence.
\begin{figure}[h!]
        \centering
         \includegraphics[width=0.4\textwidth]{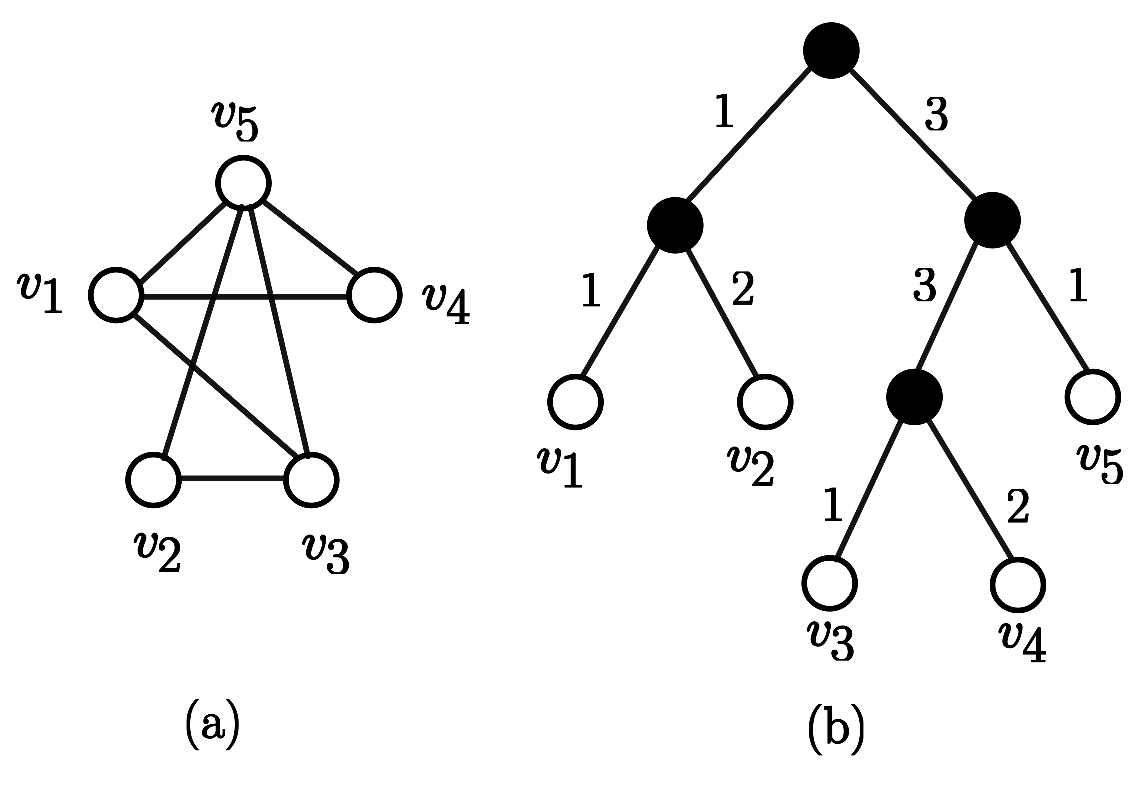}
        \caption{(a) A PCG with five vertices;
        (b) An edge-weighted tree due to which the graph in~(a) is PCG by considering the interval [5, 10] and the identity correspondence}
        \label{fig:PCG}
    \end{figure}

Kearney et al.~\cite{kearney2003efficient} conjectured that every graph is a PCG. The conjecture was then refuted by Yanhaona et al.~\cite{yanhaona2010discovering}  by showing a bipartite graph with fifteen vertices that is not a PCG. 
Similarly, Durocher et al.~\cite{durocher2015graphs} refuted the conjecture by showing a non-planar graph with eight vertices and a planar graph with sixteen vertices that are not PCGs. 
 Phillips~\cite{phillips2002uniform}  exhaustively established that all graphs with less than five vertices are PCGs. 
 This result is further extended by Calamoneri et al.~\cite{calamoneri2013all}  by demonstrating that all graphs with at most seven vertices are PCGs. 
  Xiao and Nagamochi~\cite{xiao2020some}  characterized PCGs and proved that a graph is a PCG if and only if each of its bi-connected components is a PCG. 
 Hossain et al.~\cite{hossain2016necessary}  developed a necessary and sufficient condition for a graph to be a PCG. 
Azam et al.~\cite{azam2018enumerating, azam2021method} proposed a method based on linear programming to enumerate PCGs, and proved that there are exactly seven graphs with eight vertices that are not PCGs. 
 Subsequently, Azam et al.~\cite{azam2022enumeration}  improved the method and enumerated all minimal non-PCGs (non-PCGs whose all proper induced subgraphs are PCG)  with nine vertices. 
 More precisely they proved that there are exactly 1,494 minimal non-PCGs with nine vertices. The number  of minimal non-PCGs with at least 10 vertices is unknown. 
 
Ahmed and Rahman~\cite{ahmed2017multi} introduced a natural generalization of PCGs, called {\em multi-interval PCGs/$k$-interval PCGs ($k$-IPCGs)}. The class of $k$-IPCGs is designed to permit $k$ mutually exclusive intervals of non-negative real numbers rather than just one interval. Ahmed and Rahman~\cite{ahmed2017multi} proved that every graph is a $k$-IPCG for some $k$.
Clearly, every PCG is a $k$-IPCG, but the converse is not true. For example, all wheel graphs with at least nine vertices are not PCGs~\cite{baiocchi2019some}, but are 2-IPCGs~\cite{ahmed2017multi}.
Furthermore, it is known that not every graph is a 2-IPCG. The smallest known graph that is not a 2-IPCG consists of 135 vertices~\cite{calamoneri2021generalizations}.
Calamoneri et al.~\cite{calamoneri2022all}  analytically showed that each of the seven non-PCGs reported by Azam et al.~\cite{azam2021method}  with eight vertices are 2-IPCGs. They further observed that the number of non-PCGs on nine vertices is much bigger than that on eight vertices, and therefore their approach cannot be applied to test if a given graph is a 2-IPCG with at least nine vertices~\cite{calamoneri2022all}. The number and structure of 2-IPCGs with at least nine vertices are unknown.
A summary of the current status of known/unknown PCGs and 2-IPCGs with up to 10 vertices is given in Table~\ref{tab:summaryPCGs}. 
Note that to prove if a given graph is a $k$-IPCG or not, one needs to search over an infinite space of weights and intervals, and an exponentially large number of trees and correspondences. Moreover, the number of graphs increases exponentially with the number of vertices, making the problem more difficult.

In this paper, we propose a method to enumerate $k$-IPCGs with a given number of vertices. 
The idea of our method is to generate almost all $k$-IPCGs for a given tree by randomly assigning weights and adjusting intervals in a typical way. 
To avoid the generation of isomorphic $k$-IPCGs, we use canonical representations proposed by McKay and Piperno~\cite{mckay2014practical}. 
To reduce the exponentially large tree search space, we theoretically prove that every $k$-IPCG can be constructed by using a binary tree, a tree whose all internal vertices have degree exactly three. 
The proposed method is implemented and applied to enumerate all $2$-IPCGs with up to ten vertices. 
With computational results, we prove that all graphs with up to ten vertices are $2$-IPCGs. 
\begin{table}[ht]
\centering
\begin{tabular}{|c|c|c|c|}
\hline
\#vertices & \#graphs & Known/Unknown PCGs & Known/Unknown 2-IPCGs \\
\hline
1 & 1 & Known & Known   \\
2 & 2 & Known& Known \\
3 & 4 & Known & Known  \\
4 & 11 & Known& Known  \\
5 & 34 & Known& Known \\
6 & 156 & Known & Known \\
7 & 1,044 & Known & Known  \\
8 & 12,346 & Known & Known \\
9 & 274,668 & Known & \bf{Unknown}  \\
10 & 12,005,168 & \bf{Unknown} & \bf{Unknown}  \\
\hline
\end{tabular}
\caption{Summary of known/unknown PCGs and 2-IPCGs}
\label{tab:summaryPCGs}
\end{table}

The structure of the paper is as follows:  Preliminaries and related results are discussed in Section~\ref{sec: preliminaries}. 
A relationship between $k$-IPCGs and binary trees, and enumeration of binary trees is discussed in Section~\ref{sec:binary}. 
A method to generate $k$-IPCGs is described in Section~\ref{sec:genPcg}. 
Experimental results are discussed in Section~\ref{sec: result}, and a conclusion with future directions is provided in Section~\ref{sec:conclusion}. 
\section{Preliminaries} \label{sec: preliminaries}
For a simple and undirected graph $G$, we denote by $V(G)$ and $E(G)$ the vertex set and the edge set, respectively. 
An edge connecting the vertices $u$ and $v$ is represented by $uv$. 
If $e = uv \in E(G)$, then the vertices $u$ and $v$ are said to be adjacent, and the edge $e$ is incident to both $u$ and $v$. 
The {\em degree} of a vertex is defined to be the number of edges that are incident to it.  
A labeled graph with the natural edge-labeling is called a \textit{natural-labeled graph}. 
A natural-labeled graph on $n$ vertices with the vertex label set $\{1, 2, \ldots, n\}$ is said to be a \textit{basic-labeled graph}. 
Henceforth, for a simple natural-labeled graph, and a simple basic-labeled graph, we simply use the terms graph and basic-labeled graph, respectively, unless stated otherwise.

Let $h$ be a non-negative integer such that $h \leq |A|-1$, where $A$ is a set. 
We define a \textit{$k$-coloring} of $A$ to be a function from $A$ to the set 
$\{0,1,\ldots, h - 1\}$ of colors. 
Let $\Lambda_{h}(A)$ denote the set of all $h$-colorings of $A$. 
A graph with an $h$-vertex coloring or $h$-edge coloring is said to be a \textit{colored graph}.
We define a \textit{full $h$-vertex coloring} (resp., \textit{full $h$-edge coloring)} of $G$ to be a $h$-vertex coloring (resp., $h$-edge coloring) such that the range set is $\{0,1, \ldots, h- 1\}$.

Let $G$ and $G'$ be two graphs. These graphs are said to be {\em isomorphic} if there exists a bijection $\sigma: V(G) \rightarrow V(G')$ which preserves the edges, i.e., $uv \in E(G)$ if and only if $\sigma(u)\sigma(v) \in E(G')$. 
In this case, we write $G \cong G'$.
Two vertex-colored (resp., edge-colored) graphs $(G,\pi)$ and $(G', \pi')$ (resp., $(G,\rho)$ and $(G', \rho')$) 
are said to be isomorphic as vertex-colored (resp., edge-colored) graphs if there exists an edge preserving bijection $\sigma: V(G) \to V(G')$ that preserves vertex (resp., edge) color, i.e., for each $u\in V(G)$ (resp. $uv \in E(G)$) it holds that $\pi(u) = \pi'(\sigma(u))$  (resp., $\rho(uv) = \rho'(\sigma(u)\sigma(v))$), in which case we write $(G,\pi) \cong_{\textnormal{v}} (G', \pi')$ (resp., $(G,\rho) \cong_{\textnormal{e}} (G', \rho')$).
In all of these cases, we call $\sigma$ an \textit{isomorphism} between the given graphs. 
An isomorphism from $G$ to itself is called an \textit{automorphism} of the graph.
Let $\mathcal{G}_n$ denote the set of all mutually non-isomorphic graphs with $n$ vertices. 

Let $\mathcal{B}_n$ denote the set of all basic-labeled graphs with $n$ vertices. 
Let $H$ denote the set $\{1,2, \ldots, n\}$. 
We define a \textit{canonical representation function} $C:\mathcal{B}_n \times \cup_{h = 0} ^{n-1}\Lambda_{h}(H) \to \mathcal{B}_n \times \cup_{h = 0} ^{n-1}\Lambda_{h}(H)$ such that 
\begin{enumerate}
\item[(i)] For a graph $G \in \mathcal{B}_n$ and a coloring $\pi \in \cup_{h = 0} ^{n-1}\Lambda_{h}(H)$ it holds that $C(G, \pi) \cong_{\textnormal{v}} (G, \pi)$; and 
\item[(ii)] For any two graphs $G, G' \in \mathcal{B}_n$ and any two colorings $\pi, \pi' \in \cup_{h = 0} ^{n-1}\Lambda_{h}(H)$ it holds that $(G,\pi) \cong_{\textnormal{v}} (G', \pi')$ if and only if $C(G, \pi) = C(G', \pi')$. 
\end{enumerate}
Intuitively, a canonical representation function (CRF) maps a given colored graph to a colored graph which is a unique representative of its isomorphism class. 
For a CRF $C$ and a colored graph $(G, \pi) \in \mathcal{B}_n \times \cup_{h = 0} ^{n-1}\Lambda_{h}(H)$, we call $C(G, \pi)$ a \textit{canonical representation} (CR) of the colored graph. For a CRF $C$, we define a \textit{hash function} $f$ which takes a canonical representation of a colored graph as input and outputs a value such that for $(G, \pi), (G', \pi') \in \mathcal{B}_n \times \cup_{h = 0} ^{n-1}\Lambda_{h}(H)$, with $f(C(G,\pi))\neq f(C(G', \pi'))$ it holds that $(G, \pi) \ncong_{\textnormal{v}} (G', \pi')$, i.e., two CRs of two isomorphic vertex colored graphs have the same hash value.

A tree is an acyclic, connected graph. 
The vertices of degree 1 (resp., greater than 1) are called \textit{leaves} (resp.,  \textit{internal vertices}) of a tree. 
A tree with a specific vertex is called a \textit{rooted tree}. 
A \textit{full binary tree} is defined to be a rooted tree in which all internal vertices except the root have degree exactly three. 
For $n\geq 2$, let $\mathcal{F}_{n}$  denote the set of all mutually non-isomorphic full binary trees with $n$ leaves.
A \textit{binary tree} is defined to be a tree in which each internal vertex has degree exactly three. 
Let $\mathcal{T}_{n}$ denote the set of all mutually non-isomorphic binary trees with $n$ leaves. 
An illustration of a binary tree with nine leaves is shown in Fig.~\ref{fig:binarytree}. 
A {\em path} between two vertices $u, v$ is a sequence of vertices with end points $u$ and $v$, and each consecutive pair of vertices forms an edge.
Between any two vertices $u, v \in V(T)$ there exists a unique path, and the set of edges on this path is denoted by $E_T(u,v)$.  
A pair  $(T, w)$ of a tree $T$ and a function  $w: E(T) \rightarrow \mathbb{R}^{+}$
is called an {\em edge-weighted tree}. 
For any two vertices $u, v$ in an edge-weighted tree $(T, w)$, the {\em distance} $d_{T, w}(u,v)$ is defined to be the sum of weights of the edges in the unique path between $u$ and $v$ in $T$, i.e., 
$d_{T, w}(u,v) \triangleq \sum_{ e \in E_T(u,v)} w(e)$. 
\begin{figure}[h!]
        \centering
         \includegraphics[width=0.4\textwidth]{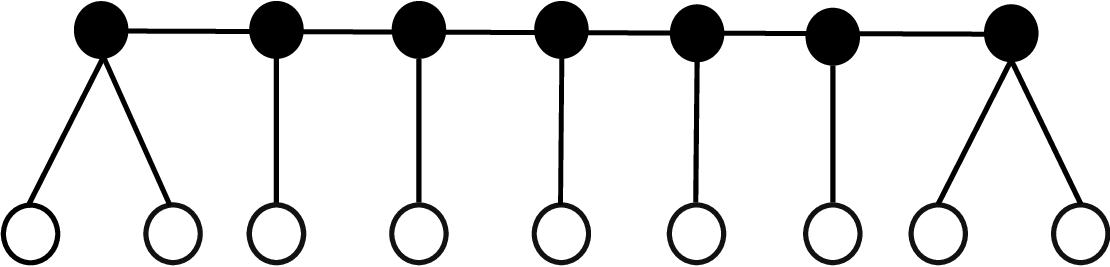}
        \caption{A binary tree with nine leaves}
        \label{fig:binarytree}
    \end{figure}
     
Let $d_{\rm {min}} , d_{\rm {max}} \in \mathbb{R}^{+}$. 
A graph $G$ is said to be a PCG if there exists an edge-weighted tree $(T, w)$, and 
a bijection $\sigma: L(T) \to V(G)$ such that 
there is an edge $uv \in E(G)$ if and only if $d_{T, w}(\sigma(u),\sigma(v)) \in 
[d_{\rm {min}}, d_{\rm {max}}]$.
$k$-IPCGs are a generalization of PCGs where $k$ disjoint intervals are used. 
More precisely, a graph $G$ is said to be a $k$-IPCG if there exist $k$ disjoint intervals 
$I_i = [d^i_{\rm {min}}, d^i_{\rm{max}}], i = 1, \ldots, k$, such that 
$d_{T, w}(\sigma(u),\sigma(v)) \in I_i$. 
In such a case, we call $T$ a {\em witness tree} and 
denote the $k$-IPCG by $k$-IPCG$(T, w, \sigma, I_{1}, \ldots, I_{k})$. An example of 2-IPCG is shown in Fig.~\ref{fig:2-IPCG}. 
\begin{figure}[h!]
        \centering
         \includegraphics[width=0.4\textwidth]{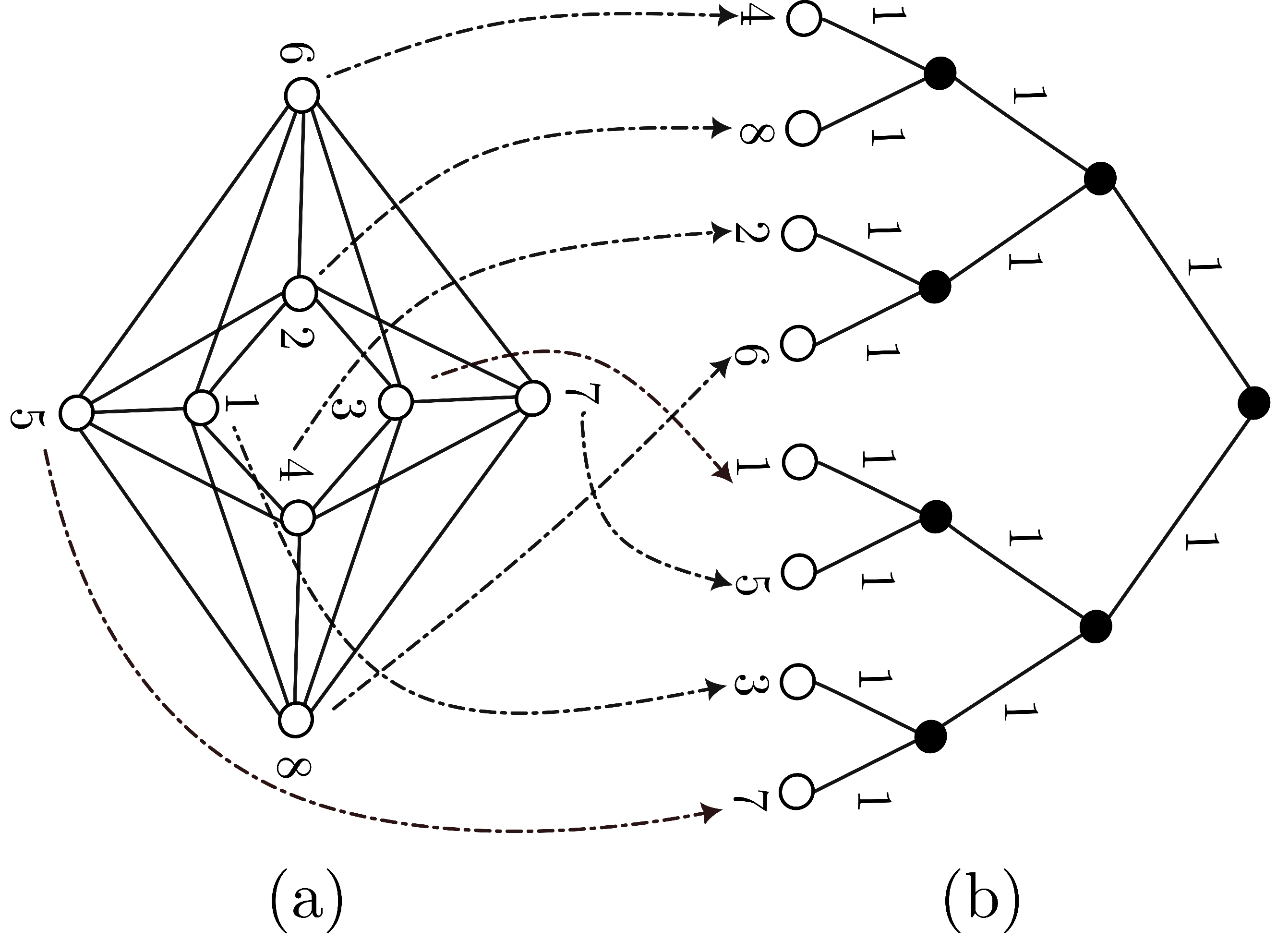}
        \caption{(a) A non-PCG proved by Azam et al.~\cite{azam2018enumerating} which is a 
        2-IPCG identified by Calamoneri et al.~\cite{calamoneri2022all}  with its witness tree in~(b) and corresponding intervals $I_{1}=[1,3], I_{2}=[5,6]$ }
        \label{fig:2-IPCG}
    \end{figure} 
    
McKay and Piperno~\cite{mckay2012nauty} developed a software called {\tt NAUTY} to compute the CR of colored graphs. The input and output of {\tt NAUTY} are given below. \\
\textbf{Input: } A basic-labeled graph $G$ with $n$ vertices and a full $k$-vertex coloring $\pi$ of $G$.\\
\textbf{Output:} 
\begin{enumerate}[noitemsep]
\item[(i)] A canonical representation $C(G, \pi)$  of the graph $(G, \pi)$ with an isomorphism between the graphs $(G,\pi)$ and $C(G,\pi)$; and 
\item[(ii)] a hash value $f(C(G,\pi)) \in [0, 2^{32} - 1]^3$ which is an ordered triplet of integers each in the range $[0, 2^{32} - 1]$.
\end{enumerate}

\subsection{Some known results}
In this section, we review some known results related to PCGs  and 2-IPCGs.

Azam et al.~\cite{azam2018enumerating} enumerated all PCGs with eight vertices.
\begin{lemma}~\cite{azam2018enumerating}
There exist exactly seven non-PCGs with eight vertices.
\end{lemma} 
Azam et al.~\cite{azam2022enumeration} improved the enumeration method and generated all minimal non-PCGs with nine vertices.
\begin{lemma}
\cite{azam2022enumeration}
There exist exactly 1,494 minimal non-PCGs with nine vertices.
\end{lemma} 
Calamoneri et al.~\cite{calamoneri2022all} proved that seven non-PCGs with eight vertices reported by Azam et al.~\cite{azam2018enumerating} are 2-IPCGs.
\begin{lemma}~\cite{calamoneri2022all}
All graphs with at most eight vertices are 2-IPCGs.
\end{lemma}
Calamoneri et al.~\cite{calamoneri2022all} proved that every 2-IPCG has a 
full binary tree as its witness tree.
\begin{lemma}~\cite{calamoneri2022all} If $G$ is a 2-IPCG, then there exists an edge-weighted full binary tree $(T, w)$ and two disjoint intervals $I_{1}$ and $I_{2}$ such that $G=$ 2-IPCG ($T, w, \sigma, I_{1}, I_{2}$). 
\end{lemma}
Calamoneri et al.~\cite{calamoneri2022all} proved that for each 2-IPCG there exist integer weights and intervals with integer end points.
\begin{lemma}~\cite{calamoneri2022all} 
Let $G$ be a 2-IPCG$(T, w, \sigma, I_1, I_2)$. 
Then there exists an integer  weight function $w'$, two disjoint intervals $I'_{1}$ and $I'_{2}$ with integer endpoints such that $G=\text{2-IPCG}(T, w', \sigma, I'_1, I'_2)$.
\end{lemma}
Ahmed and Rahman~\cite{ahmed2017multi}  discussed the wheel graphs. 
\begin{lemma}~\cite{ahmed2017multi}
All wheel graphs with at least nine vertices are 2-IPCGs.
\end{lemma}
\section{Proposed method to enumerate binary trees}\label{sec:binary}
We prove a relationship between $k$-IPCGs and binary trees in Lemma~\ref{lem:binary} to reduce the tree search space. 
Furthermore, an enumeration algorithm to generate binary trees with a given number of leaves is proposed in this section. 

 \begin{lemma}\label{lem:binary}
A graph $G$ is a $k$-IPCG if and only if there exists a binary tree $T$ such that 
$G = k$-IPCG$(T,w, \sigma, I_1, \ldots, I_k)$.
\end{lemma}
\begin{proof}
Let $(T, w)$ be an edge-weighted  binary tree and $I_{i}, i=1, \ldots k$ be $k$ disjoint intervals. 
Then the graph $k$-IPCG $(T,w, \sigma, I_1, \ldots, I_k)$ is a 
$k$-IPCG by the definition of IPCGs. 

Conversely, suppose that $G$ be a $k$-IPCG with tree $T$, weight $w$, bijection $\sigma$, and intervals  $I_1, \ldots, I_k$. 
If $T$ is a binary tree, then there is nothing to prove. 
If $T$ is not a binary tree, then $T$ will have at least one non-leaf vertex of degree smaller or greater than three. 
Suppose that $v \in V(T) \setminus L(T)$ such that ${\rm deg}(v) \neq 3$. 
If ${\rm deg}(v) < 3$, then ${\rm deg}(v)= 2$ since $v \notin L(T)$. 
Let $u, t \in V(T) \setminus \{v\}$ such that $uv, vt \in E(T)$, and 
$w(uv)$ and $w(vt)$ are the corresponding weights. 
Then delete the vertex $v$ and edges $uv$ and $vt$, and introduce a new edge $ut$ with $w(ut)= w(uv)+w(vt)$ as depicted in  Fig.~\ref{fig:Binary_tree}$(a)$. 
By this transformation, we can remove all the vertices of degree 2 from $T$.
If deg $(v)>3$ and the neighbors of $v$ are $t$, $u_{1}, \ldots, u_{\ell}$ as shown in 
Fig.~\ref{fig:Binary_tree}$(b)$. 
Then by adding new vertices $u'_{i}$, $2 \leq i \leq \ell$ and edges $vu'_{i}$, $u_{i}u'_{i}$, $1 \leq i \leq \ell$ with $w(vu'_{i})=0$ and $w(u_{i}u'_{i})=w(vu_{i})$, 
$T$ can be converted into a binary tree.
Furthermore, it is easy to observe that these transformations will not change the distance between any pairs of leaves in $L(T)$, and therefore $G$ is a $k$-IPCG due to the modified tree $T$, weights $w$, bijection $\sigma$, and intervals $I_1, \ldots, I_k$.
\begin{figure}[h!] 
\centering
   \includegraphics[width=0.6\textwidth]{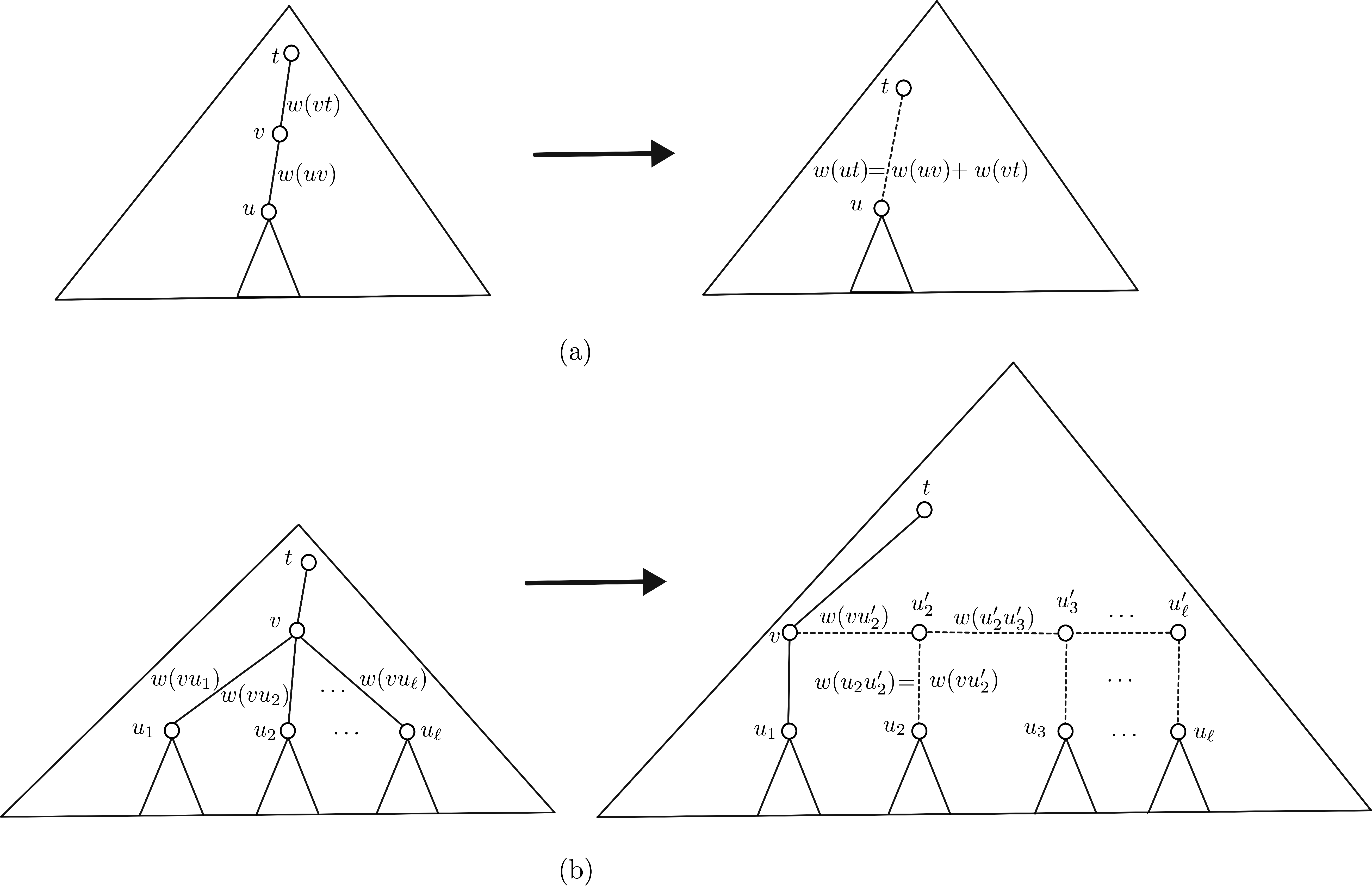}
  \caption{Transformations to convert a non-binary tree into a binary tree}
    \label{fig:Binary_tree}
    \end{figure}
\end{proof} 

Lemma~\ref{lem:binary} implies that every $k$-IPCG has a binary tree as its witness tree. 
Since the number of binary trees is much smaller than the number of trees with a given number of leaves, and therefore the search space for testing  if a given graph is a $k$-IPCG or not can be reduced by only considering binary trees. 
We next propose an algorithm to generate all non-isomorphic binary trees  for a given number of leaves. 
The idea of our algorithm is to generate a binary tree by combining some valid full binary trees at a root.  
For this purpose, we first generate full binary trees as follows. 

For any two rooted trees $T$ and $T'$ with roots $r_{T}$ and $r_{T'}$, we define 
$(r, T, T')$ to be a tree rooted at $r$, and trees $T$ and $T'$ are linked to $r$ through their roots $r_{T}$ and $r_{T'}$. 
An illustration of such a tree is given in Fig.~\ref{fig:rootedtree}, where the newly added edges are depicted by dashed lines.
\begin{figure}[h!] 
        \centering
         \includegraphics[width=0.3\textwidth]{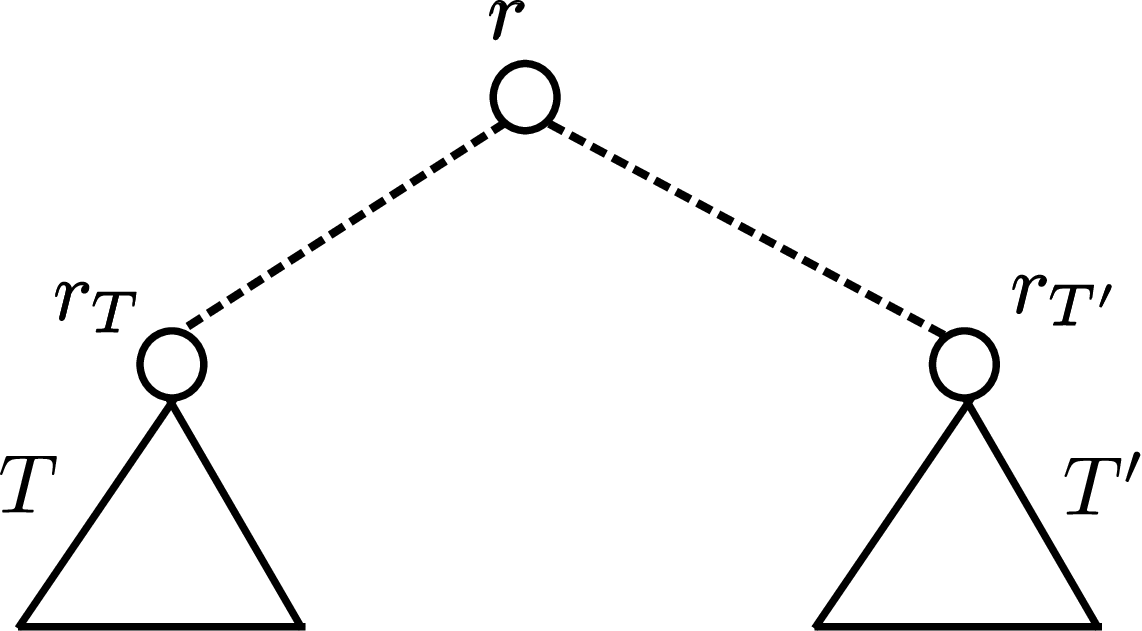}
        \caption{The rooted tree $(r, T, T')$ obtained from $T$ and $T'$}
        \label{fig:rootedtree}
    \end{figure}

Observe that $(r, T, T')$ is a full binary tree with $n$ leaves if $T$ and $T'$ are full binary trees with $n_{1}$ and  $n_{2}$ leaves, reps., and $n_{1}+ n_{2}=n$. Based on this observation, we develop Algorithm~\ref{alg:fullbin},
to generate 
the set $\mathcal{F}_{n}$ of all mutually non-isomorphic full binary trees with $n$ leaves. 
\begin{algorithm}[htbp]
\caption{GenFBTree($n$)} \label{alg:fullbin}
\begin{algorithmic}[1]
\item[\textbf{Input:}] An integer $n\geq2$.
\item[\textbf{Output:}] The family $\mathcal{F}_{n}$.
\State $\mathcal{F}:= \phi$;
\For{each sequence 
$(n_{1}, n_{2})$ with $n_{1} \geq n_{2}$ and $n_{1}+n_{2}=n$ }
\For{$T_{i} \in \mathcal{F}_{{n}_{1}}$}
\For{$T_{j} \in \mathcal{F}_{{n}_{2}}$}
\If{$n_{1}= n_{2}$} \hspace{1cm} /* only possible when $n$ is even */
\If {$i \geq j$}
\State $\mathcal{F}:=\mathcal{F} \cup \{(r, T_{i}, T_{j})\} $
\EndIf
\Else
\State $\mathcal{F}:=\mathcal{F} \cup \{(r, T_{i}, T_{j})\} $ 
\EndIf
\EndFor
\EndFor
\EndFor;
\State Return $\mathcal{F}$ as $\mathcal{F}_{n}.$
   \end{algorithmic}
   \end{algorithm}

The correctness of Algorithm~\ref{alg:fullbin} is established in Lemma~\ref{lem:cor_full}. 
\begin{lemma}\label{lem:cor_full}
Algorithm~\ref{alg:fullbin} outputs all non-isomorphic full binary trees for a given $n$.
\end{lemma}
\begin{proof}
It is easy to observe that each full binary tree can be viewed as $(r, T_{i}, T_{j})$,  where $T_{i} \in \mathcal{F}_{{n}_{1}}$ and $T_{j} \in \mathcal{F}_{{n}_{2}}$ with the condition that $n_{1}+n_{2}=n$. 
Therefore due to the for-loop at line~2 of Algorithm~\ref{alg:fullbin}, GenFBTree($n$) covers all the required trees. 
Furthermore, the conditions $n_{1} \geq n_{2}$ and ``$i \geq j  $ when $n_{1} = n_{2}$'' ensure that all the generated trees will be non-isomorphic since only the root has degree 2.
\end{proof}

Each tree can be uniquely viewed as a rooted tree by considering either its unicentroid or bicentroid as the root~\cite{jordan1869assemblages}. 
It is important to note that the bicentroid case applies when the number of vertices is even. 
Based on this observation, we propose an algorithm for enumerating binary trees as follows. 

We define $(r,T, T', T'')$ to be a rooted tree with the unicentroid obtained by attaching rooted trees $T$, $T'$, $T''$ to $r$ through their roots $r_{T}$, $r_{T'}$ and $r_{T''}$.  Similarly, $(e= r_{T}r_{T'}, T, T')$ is defined to be a rooted tree with a bicentroid obtained by connecting  $T$ with $T'$ by creating an edge $e$ between $r_{T}$ and $r_{T'}$, where $|V(T)| =|V(T')|$. 
Illustrations of these trees are given in Figs.~\ref{fig:rootedtrees}(a) and (b).
\begin{figure}[h!] 
        \centering
          \includegraphics[width=0.7\textwidth]{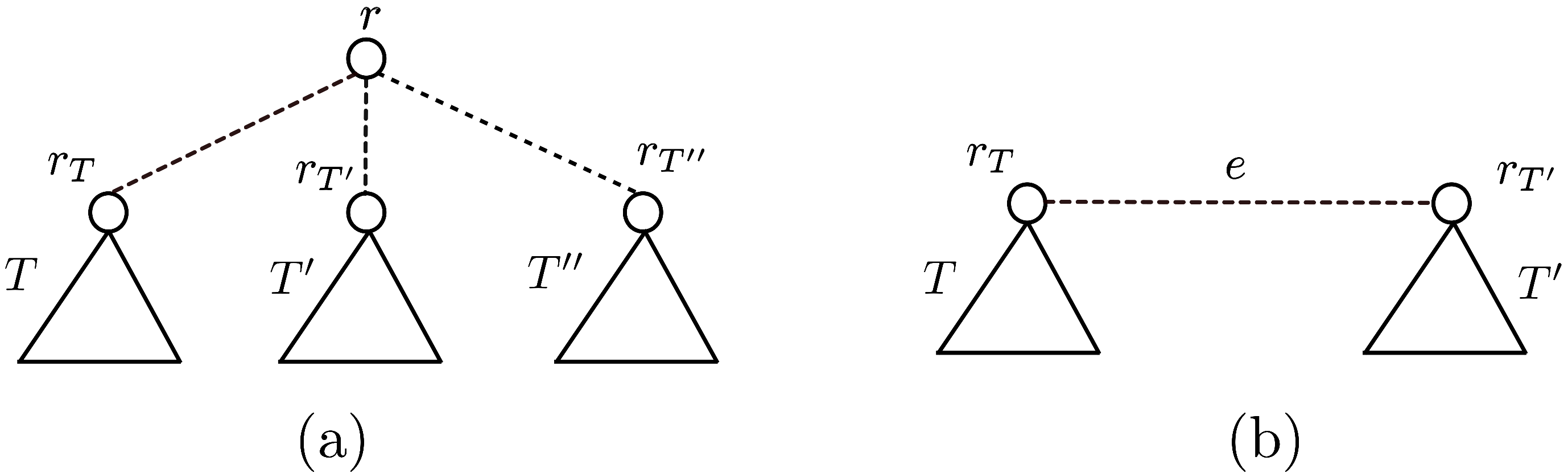}
        \caption{
        (a)~The rooted tree $(r,T, T', T'')$ with unicentroid obtained from $T$, $T'$ and $T''$, where the edges between roots of the trees are depicted by dashed lines; 
        (b)~The rooted tree $(e, T, T')$ with bicentriod obtained from $T$ and $T'$, where the edge between roots is depicted by a dashed line}
        \label{fig:rootedtrees}
    \end{figure}

Observe that $(r,T, T', T'')$ is a binary tree with $n$ leaves if $T$, $T'$, $T''$ are full binary trees with $n_{1}$, $n_{2}$, $n_{3}$  leaves, resp., and $n_{1}+ n_{2}+ n_{3}=n$. 
Similarly $(e, T, T')$ is a binary tree  with $n$ leaves if $T$, $T'$ are full binary trees with $n_1, n_2$ leaves, resp., and $n_{1}= n_{2}$.
From this observation, Algorithm~\ref{alg:bin} is designed to generate 
the set $\mathcal{T}_n$ of mutually non-isomorphic binary trees with $n$ leaves. 
\begin{algorithm}[htbp]
\caption{GenCan($\mathcal{G}$)}\label{alg:gencan}
\begin{algorithmic}[1] 
\item[\textbf{Input:}] A non-empty set $\mathcal{G}$ of basic-labeled graphs with $n$ vertices.
\item[\textbf{Output:}] Generate canonical representations and hash values of each graph in $\mathcal{G}$.
\State $\mathcal{C}:=\phi$;
\For {$G$ $\in$ $\mathcal{G}$ }
\State Compute $C(G, \pi_{0})$ and $h$($C(G, \pi_{0}))$ using {\tt NAUTY};
\State $\mathcal{C} := \mathcal{C}\cup \{C(G, \pi_{0}), \mathit{h}(C(G, \pi_{0}))\}$
\EndFor;
\State Return \(\mathcal{C}\) is the required set of canonical representation with hash values.
   \end{algorithmic}
   \end{algorithm}
 \pagebreak
  
\begin{algorithm}[htbp]
\caption{GenBTree$(n)$ 
} \label{alg:bin}
\begin{algorithmic}[1]
\item[\textbf{Input:}] An integer $n\geq2$.
\item[\textbf{Output:}] The family $\mathcal{T}_{n}$.
\State  $\mathcal{F}_{\ell}:=$ GenFBTree$(\ell)$, where $\ell\leq \lfloor \frac{n}{2} \rfloor$;
\State $\mathcal{T}:= \phi$;
\For{\textbf{each} sequence $(n_{1}, n_{2}, n_{3})$  with $n_{1}+n_{2}+n_{3}=n$,  $n_{1}, n_{2}, n_{3}  \geq 1 $, and $n_{1} \geq n_{2} \geq n_{3} $}
\For{$T_{i} \in \mathcal{F}_{{n}_{1}}$}
\For{$T_{j} \in \mathcal{F}_{{n}_{2}}$}
\For {$T_{k} \in \mathcal{F}_{{n}_{3}}$}
\If{$n_{1}= n_{2}$ or $n_{2}= n_{3}$ or $n_{1}= n_{2}= n_{3}$ } 
\If {``$i \geq j$, $n_{1}= n_{2}$'' or
  ``$j \geq k$, $n_{2}= n_{3}$'' or
   ``$i \geq j \geq k$, 
   \Statex \hspace{2.7cm}$n_{1}= n_{2}= n_{3}$''}
\State $\mathcal{T}:=\mathcal{T} \cup \{(r, T_{i}, T_{j}, T_{k})\} $
\EndIf
\Else
\State $\mathcal{T}:=\mathcal{T} \cup \{(r, T_{i}, T_{j}, T_{k})\} $ 
\EndIf
\EndFor
\EndFor
\EndFor
\EndFor;
\If{$n$ is even}
\For{\textbf{each} $T_{i}, T_{j} \in \mathcal{F}_{\ell}$ such that $i \geq j$ }
\State $\mathcal{T}:=\mathcal{T} \cup \{(e, T_{i}, T_{j})\} $
\EndFor
\EndIf;
\State $\mathcal{T'}:=$ GenCan($\mathcal{T}$);
\State Return $\mathcal{T'}$ as $\mathcal{T}_{n}$.
   \end{algorithmic}
   \end{algorithm}

  The correctness of Algorithm~\ref{alg:bin} is discussed in Lemma~\ref{lem:cor_bi}. 
   
  \begin{lemma}\label{lem:cor_bi}
 Algorithm~\ref{alg:bin} generates the set $\mathcal{T}_n$  of all non-isomorphic binary trees  for a given~$n$.
 \end{lemma} 
 \begin{proof}
The condition $\ell \leq \lfloor \frac{n}{2} \rfloor$ ensures that the degree of the root of the resultant trees is exactly three. The tuples $(r, T_{i}, T_{j}, T_{k})$, where $T_{i} \in \mathcal{F}_{n_{1}}$, $T_{j} \in \mathcal{F}_{n_{2}}$, and $T_{k} \in \mathcal{F}_{n_{3}}$, with the condition that $n_{1}+ n_{2}+ n_{3}=n$, cover all the binary trees with a unicentroid. The conditions $n_{1} \geq n_{2} \geq n_{3}$ and ``$i \geq j$, $j \geq k$, or $i \geq j \geq k$ when $n_{1}= n_{2}$, $n_{2}= n_{3}$, or $n_{1}= n_{2}= n_{3}$" are used to avoid the generation of some isomorphic trees. The case of rooted trees with a bicentroid is addressed by the triplets $(e, T_{i}, T_{j})$ when $n$ is even. Hence, Algorithm~\ref{alg:bin} outputs all non-isomorphic binary trees.
 \end{proof}
  \section{Proposed method to generate $k$-IPCGs}  \label{sec:genPcg}
For a given binary tree, we wish to generate almost all $k$-IPCGs whose witness tree is a given binary tree $T$. 
 Our algorithm relies on the observation that for each $k$-IPCG there exist disjoint $k$ intervals with endpoints from the distances between leaves. More precisely if $G= k$-IPCG$(T,w, \sigma, I_{1}, \ldots, I_{k})$ then by setting 
 $I'_{i}=[x_i, y_i]$ such that 
 $x_i ={\rm min}\{d_{T,w}(a,b)|d_{T,w}(a,b) \in  I_i\}$ and 
 $y_i ={\rm max}\{d_{T,w}(a,b)|d_{T,w}(a,b) \in  I_i\}$, 
 we have $G=k$-IPCG$(T,w, \sigma, I'_{1}, \ldots, I'_{k})$. 
 This observation allows us to generate all $k$-IPCGs for a fixed binary tree $T$ by  randomly assigning $w$ and setting $k$ intervals with endpoints from the distances between the leaves. 
Selecting such intervals can be done easily by listing all distinct distances in strictly increasing order.
In these settings, the number of edges in the resultant $k$-IPCGs can be determined before their construction, and hence can be used to determine if $k$-IPCGs with that number of edges need to be generated or not. 
That is, if all $k$-IPCGs with a certain number of edges are already generated, then there is no need to generate $k$-IPCGs of the same size again. 
Finally, to efficiently apply isomorphic test, we partition all the graphs with respect to their number of edges.  
Note that the proposed method can generate $k$-IPCGs with at least three edges.
The main steps of the proposed generator are outlined below:
 \begin{enumerate} 
 \item Randomly generate a weight sequence of length $2n-3$ as there are $2n-3$ edges in a binary tree with $n$ leaves.
 \item For each tree $T \in \mathcal{T}_n$ assign weights to the edges.
 \item Compute the distances between each pair of leaves in the binary trees $T$.
 \item Select endpoints of $k$ intervals from the distinct distances.
 \item Construct $k$-IPCGs by assuming identity bijection.
 \item Finally test if the generated $k$-IPCG is a new $k$-IPCG or not by applying isomorphism test.
 \end{enumerate}
 
An algorithmic description of the proposed method is given in the Algorithm~\ref{alg:genkpcg}. 
\begin{algorithm}[htbp]
\caption{Partition($\mathcal{G}$)} \label{alg:partition}
\begin{algorithmic}[1] 
\item[\textbf{Input:}] A set $\mathcal{G}$ of graphs with $n$ vertices.
\item[\textbf{Output:}] Partition of graphs in $\mathcal{G}$ with respect to the number of edges. 
\State $\text{num}_{i}:=0, i=1,2, \ldots, \binom{n}{2}$;
\For {$G$ $\in$ $\mathcal{G}$ }
\If {$\text{num}_{|E|}=0$}
\State $\mathcal{P}_{|E|}:=\phi$;   \hspace{1cm} /* Initialize $\mathcal{P}_{|E|}$*/
\State $\mathcal{P}_{|E|}:=\mathcal{P}_{|E|}\cup \{G\}$ ;
\State $\text{num}_{|E|}:=1$
\Else
 \State  $\mathcal{P}_{|E|}:=\mathcal{P}_{|E|}\cup \{G\}$
\EndIf
\EndFor;
\State Return $\mathcal{P}_{i}s$ as the required partition.
   \end{algorithmic}
   \end{algorithm}

\begin{algorithm}[htbp]
\caption{GetPathEdges$(T, s, t)$} \label{alg:path}
\begin{algorithmic}[1] 
\item[\textbf{Input:}] A tree $T$, and its two vertices $s$ and $t$.
\item[\textbf{Output:}]  All edges on the path between $s$ and $t$ in $T$.
   \State Path := $\emptyset$;
    \ForAll{$u \in V(T)$}
        \State visited[$u$] := false;   preds[$u$] := $u$
    \EndFor;
    
    \State $Q$ :=  Empty queue; enqueue$(Q, s)$;
    \While{$Q$ is not empty}
       \State $v :=$ dequeue$(Q)$;
        \If{visited[$v$] = false}
            \State visited[$v$] := true
        \EndIf;
        
        \State $N[v]$ := Set of neighbors of $v$;
        \ForAll{$u  \in N[v]$}
            \If{visited[$u$] = false}
            	\State preds[$u$] := $v$; enqueue $(Q, u)$
            \EndIf
        \EndFor
    \EndWhile;
    \State $ w:=t$; \hspace{1cm}  /* {backtracking} */
    \While{preds[$w$] $\neq$ $w$}
        \State Path := Path $\cup \{{\rm preds}[w] w\}$; $w :=$  preds[$w$]
    \EndWhile;
    \State Return Path as the set of edges on the path between $s$ and $t$. 
\end{algorithmic}
\end{algorithm}
\restoregeometry 

\begin{algorithm}[htbp]
\caption{$k$-IPCGGen($\mathcal{G} ,\mathcal{T}, k, \tau$) } \label{alg:genkpcg}
\begin{algorithmic}[1]
\item[\textbf{Input:}] 
A subset $\mathcal{G}$ of $\mathcal{G}_{n}$ with $n$ vertices possibly with basic vertex labeling;
a subset $\mathcal{T}$ of tree set $\mathcal{T}_{n}$  with $n$ leaves where $n\geq 3$, leaf edges are indexed from 1 to $n$, and internal edges are indexed from $n+1$ to $2n-3$;
an integer $k\geq 1$;
a weight range $[a_{1}, a_{2}]$ for leaf edges;
a weight range $[a_{3}, a_{4}]$ for internal edges; and 
elapse time $\tau$.
\item[\textbf{Output:}]  
Until the elapse time is $\tau$, generate $k$-IPCGs that are in $\mathcal{G}$ and has a  witness tree in $\mathcal{T}$
with weight ranges $[a_{1}, a_{2}]$ and $[a_{3}, a_{4}]$. 
\State $\{\mathcal{P}_{1},\ldots, \mathcal{P}_{\binom{n}{2}}\}:=$ Partition$(\mathcal{G})$;
\State $\mathcal{C}:=$ GenCan$(\mathcal{G})$;
\State $\{\mathcal{C}_{1},\ldots, \mathcal{C}_{\binom{n}{2}}\}:=$ Partition$(\mathcal{C})$;
\State  $\mathcal{M}:=\phi$ ; \hspace{1cm} /* set of generated $k$-IPCG */
\State PathEdges$:= \phi$; \hspace{1cm} /* To store all path edges */
\For{$T_{t} \in \mathcal{T}$}
\State $PE:=\phi$;
\For{$u, v \in L(T_{t})$ with $u<v$}
\State $P:=$ GetPathEdges$(u,v)$;
\State $PE:= PE \cup \{(P,u,v,t)\}$
\EndFor;
\State PathEdges $:=$ PathEdges $\cup \{(PE,t)\}$;
\EndFor;
\While{elapsed time $<\tau$}
\State $\mathcal{G^{*}}:= \phi$; \hspace{1cm} /* store the $k$-IPCG generated in this round */
\State Randomly generate $n$ and $n-3$ weights with $w[i] \in [a_{1}, a_{2}]$, 1$\leq i\leq n$, and 
\Statex \hspace{0.6cm}$w[i] \in [a_{3}, a_{4}]$, $n+1 \leq i \leq 2n-3$;
\For{$t=1, \ldots, |\mathcal{T}|$}
\State Regard $w_{i}$, 1$\leq i\leq n$ and $n+1 \leq i \leq 2n-3$ weights of leaf edges and 
\Statex \hspace{1.12cm}internal edges of $T_{t}$, respectively;
\For{$u, v \in \{1, \ldots, n\}$ with $u<v$} 
\State $d[u,v]:= \sum_{e \in (P,u,v,t)} w[e]$
\EndFor;
\State Find distinct distances $d_{i}$, 1$\leq i\leq \ell$ s.t. $d_{i} < d_{i+1}$, $\forall i$;
\State  $S_{i}:=|\{uv$ $|$ $d[u,v]= d_{i} \}|$, 1$\leq i\leq \ell$;
\State sum$[0]:=0$;
\For{$i=1, \ldots, \ell$ }
\State sum$[i]:=$ sum$[i-1] + S_{i}$
\EndFor;
\algstore{bkbreak}
   \end{algorithmic}
   \end{algorithm}
      \addtocounter{algorithm}{-1}
   \begin{algorithm}[h]
\caption{Continued }
\begin{algorithmic}[1]
\algrestore{bkbreak}

\For{each sequence $(d_{h^{1}}, d_{h^{2}}, \ldots d_{h^{2k}} ) $ over $d_{1}, d_{2}, \ldots, d_{\ell}$ 
such that $d_{h^1} \leq$ 
\Statex \hspace{1.2cm}$d_{h^2} \leq \ldots  \leq d_{h^{2k}}$  }
\State $I_j := [d_{h^j}- \frac{1}{2} , d_{h^{j+1}}+ \frac{1}{2}]$, $1\leq j \leq k$;
\State $p:=\sum_{j=1}^{k}$( sum$[h^{j+1}]-$sum$[h^j]$ ); \hspace{0.6cm} /* no. of edges in the $k$-IPCG */
\If{$\mathcal{C}_p \neq \phi$}
\State Construct $k$-IPCG $G:=k$-IPCG$(T_{t}, w, \sigma, I_1, \ldots, I_k )$, where 
$\sigma$ is the 
\Statex \hspace{2.2cm}identity bijection 
\State $\mathcal{G}^{*}:= \mathcal{G}^{*} \cup \{G\}$
\EndIf
\EndFor
\EndFor;
\State $\mathcal{C}^{*}:=$ GenCan$(\mathcal{G}^{*})$;
\State $\{{\mathcal{C}_{1}^{*}}, \ldots, {\mathcal{C}_{n}^{*}} \}:= $ Partition$ (\mathcal{C}^{*})$;
\For{$q=1, \ldots, \binom{n}{2}$}
\For{each $(C, h) \in \mathcal{C}_{q}^{*}$}
\If{($C$, $h) \notin \mathcal{C}_{q} $}
\State $\mathcal{M} := \mathcal{M} \cup \{C, w,t, I_1, \ldots, I_k\}$ where
$C$ is the canonical representa-
\Statex \hspace{2.2cm}tion
\Statex \hspace{7em}  of $k$-IPCG$(T_{t}, w, \sigma, I_1, \ldots, I_k)$;
\State $ \mathcal{C}_{q}:=  \mathcal{C}_{q}\setminus\{( C, h)\}$
\EndIf
\EndFor
\EndFor
\EndWhile;
\State Return the set $\mathcal{M}$ of generated $k$-IPCGs, and the set 
$\mathcal{C}$ of canonical representations of the graphs in $\mathcal{G}$ that are not identified as $k$-IPCGs in the given time .
\end{algorithmic}
\end{algorithm}
\restoregeometry

 A flowchart with an example of the proposed algorithm for $k=2$ is given in Fig.~\ref{fig:kPCGs}.
 \begin{figure}[h!] 
        \centering
         \includegraphics[width=0.6\textwidth]{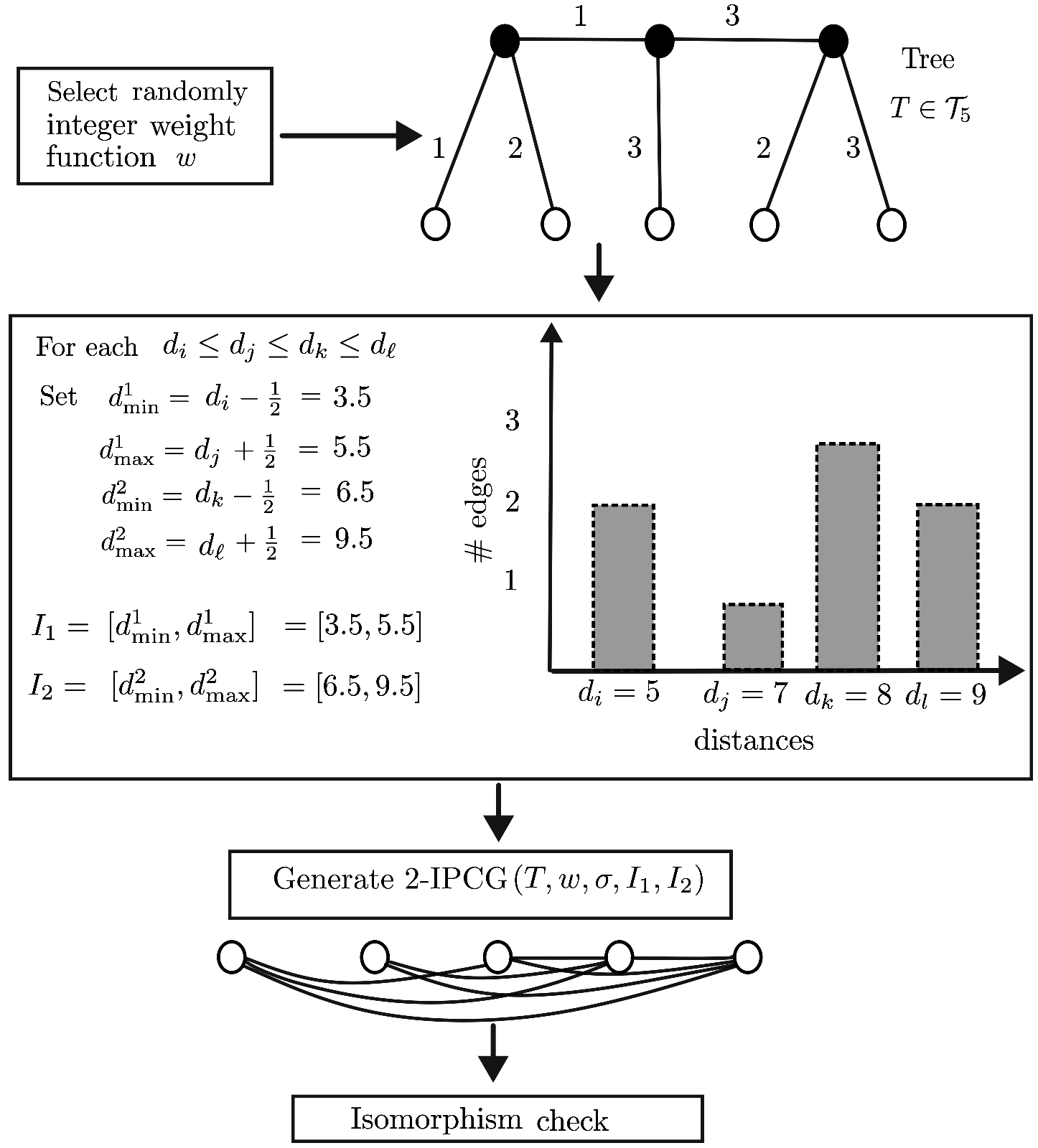}
        \caption{A flowchart of generating 2-IPCGs}
        \label{fig:kPCGs}
    \end{figure}
\section{Results and discussion} \label{sec: result}
The proposed algorithm is implemented and tested to generate 2-IPCGs with eight, nine and ten vertices by using a machine with Processor: 12th Gen, Core i7(1700 MHz) and Memory: 16 GB RAM.
The effect of different ranges of the weights, number of 2-IPCGs generated during each round, and the number of 2-IPCGs generated by each binary tree
are also analyzed in the following sections.
\subsection{Experimental results for graphs with eight vertices}
There are 12,346 graphs with eight vertices, and four binary trees in $\mathcal{T}_{8}$ that are enumerated by the proposed algorithm. 
These trees are illustrated in Fig.~\ref{fig:Binarytree8leaves}.
\begin{figure}[h!]
        \centering
      \includegraphics[width=.5\textwidth]{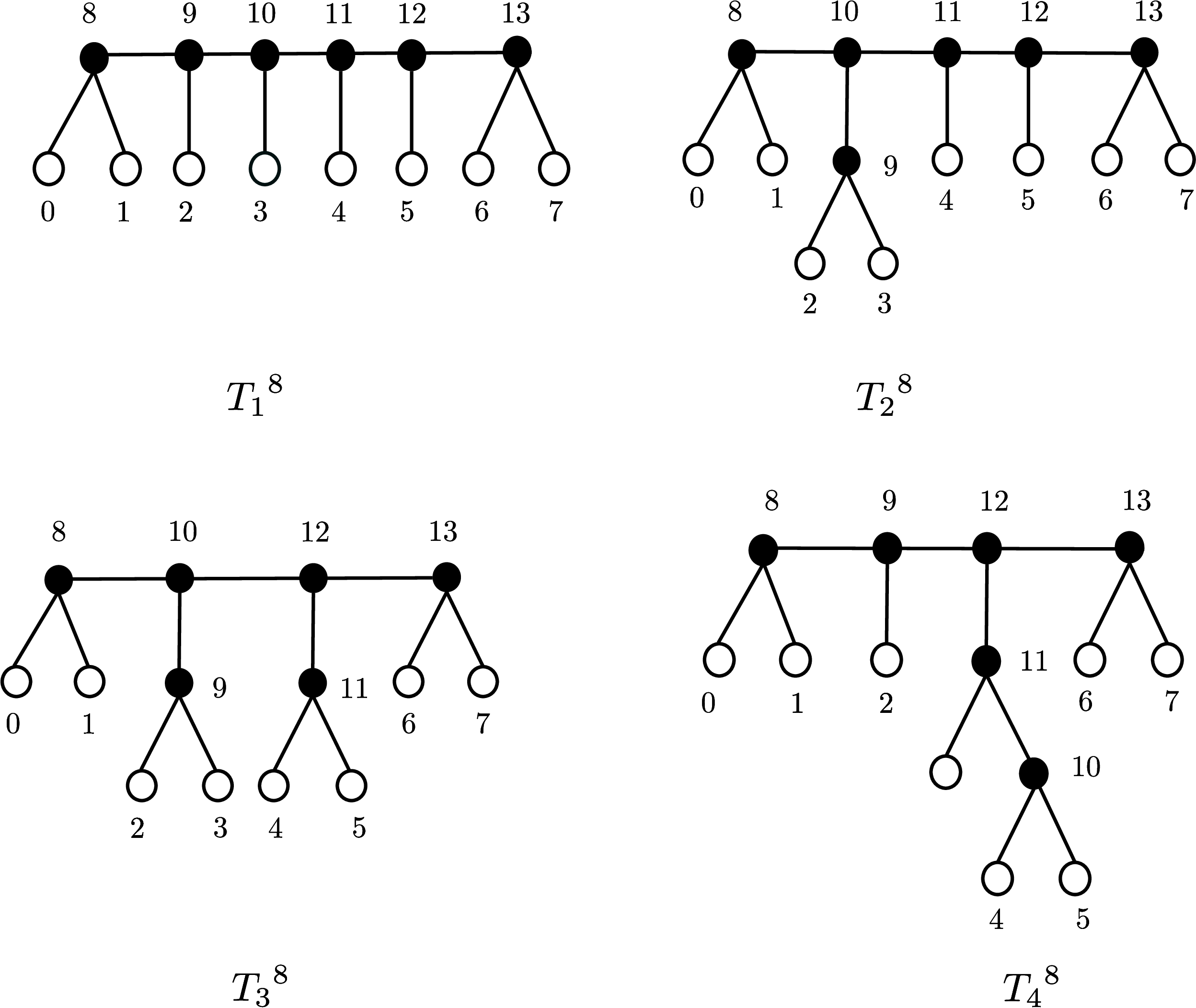}
        \caption{Binary trees with eight leaves}
        \label{fig:Binarytree8leaves}
   	 \end{figure}
   	 
2-IPCGs are generated by using all the trees in $\mathcal{T}_8$, setting two weight ranges for leaf edges $[a_{1}, a_{2}]=[1, 20]$ (resp., 
$[a_{1}, a_{2}]=[1, 50]$) and for internal edges $[a_{3}, a_{4}]=[1, 50]$
 (resp., $[a_{3}, a_{4}]=[1, 100]$).  
The enumeration results are summarized in Tables~\ref{table:computation_time_8_vertices1} and~\ref{table:time_12_minutes}.
All graphs with eight vertices are identified as 2-IPCGs in 700-815 seconds. 
From Table~\ref{table:computation_time_8_vertices1}, 
we can notice some effect of different ranges of weights, e.g., 
the number of rounds and the time to identify all 2-IPCGs are smaller for the  
range $[a_{1}, a_{2}]=[1, 50]$ and $[a_{3}, a_{4}]=[1, 100]$ as compared to the other range. 
The number of 2-IPCGs generated by each tree are listed in Table~\ref{table:time_12_minutes}.
From Table~\ref{table:time_12_minutes}, 
$T_{1}^{8}$ generated the largest number of 2-IPCGs followed by 
$T_{2}^{8}$, whereas $T_{3}^{8}$ and $T_{4}^{8}$ generated nearly the same number of 2-IPCGs for the ranges $[1, 20], [1, 50]$. 
However, the tree $T_{1}^{8}$ (resp., $T_{2}^{8}$) generated significantly large (resp., small) number of 2-IPCGs for the ranges $[1, 50], [1, 100]$.
Furthermore, the number of rounds increases over time, as the number of newly generated 2-IPCGs tends to decrease.

Calamoneri et al.~\cite{calamoneri2022all} also proved that all graphs with eight vertices are 2-IPCGs. This result confirms the correctness of our proposed generator.   
With these computational results, we have Theorem~\ref{the:8ver}.
\begin{theorem}\label{the:8ver}
All graphs with eight vertices are $2$-IPCGs. 
\end{theorem}
As a computational proof of Theorem~\ref{the:8ver}, we provide the witness trees and weights due to which a graph is a 2-IPCG  at 
\href{https://drive.google.com/drive/folders/1hGpopD12m2g6CC0Ps4a47wU7-jnMts-P?usp=drive_link}{evidence}.

\begin{table}[ht]
\centering
\begin{tabular}{|c|c|c|c|c|}
\hline
\multicolumn{5}{|c|}{$|\mathcal{G}|$ = 12,344} \\
\multicolumn{5}{|c|}{Weights $[a_1, a_2]$ = [1, 20] and $[a_3, a_4]$ = [1, 50]} \\
\hline
Sr.no. & Time (sec.) & \#Rounds & \#2-IPCGs & \#Graphs Left \\
\hline
1 & 100 & 2 & 11,426 & 918 \\
2 & 200 & 3 & 11,987 & 357 \\
3 & 300 & 5 & 12,233 & 111 \\
4 & 400 & 7 & 12,282 & 62 \\
5 & 500 & 9 & 12,307 & 37 \\
6 & 600 & 10 & 12,318 & 26 \\
7 & 700 & 12 & 12,324 & 20 \\
8 & 800 & 22 & 12,343 & 1 \\
9 & 815 & 23 & 12,344 & 0 \\
\hline
\multicolumn{5}{|c|}{Weights $[a_1, a_2]$ = [1, 50] and $[a_3, a_4]$ = [1, 100]} \\
\hline
1 & 100 & 2 & 11,810 & 534 \\
2 & 200 & 3 & 12,111 & 233 \\
3 & 300 & 4 & 12,235 & 109 \\
4 & 400 & 6 & 12,319 & 25 \\
5 & 500 & 8 & 12,333 & 11 \\
6 & 600 & 11 & 12,340 & 4 \\
7 & 700 & 12 & 12,344 & 0 \\
\hline
\end{tabular}
\caption{Computational results of 2-IPCGs with eight vertices}
\label{table:computation_time_8_vertices1}
\end{table}
\begin{table}[ht]
\centering
\begin{tabular}{|c|c|c|}
\hline
 \multicolumn{2}{|l|}{Weights  [1, 20], [1, 50] }& [1, 50], [1, 100] \\\hline
Tree & \multicolumn{2}{|c|}{\#2-IPCGs} \\
\hline
${T_{1}}^8$ & 4,227 & 6,467 \\
${T_{2}}^8$ & 3,817 & 1,319 \\
${T_{3}}^8$ & 2,136 & 2,449 \\
${T_{4}}^8$ & 2,164 & 2,109 \\
\hline
\end{tabular}
\caption{Number of 2-IPCGs generated by different trees}
\label{table:time_12_minutes}
\end{table}
\subsection{Experimental results for graphs with nine vertices}
There are 27,4668 graphs with nine vertices and 
six binary trees in $\mathcal{T}_{9}$ which are enumerated by the proposed algorithms. 
These trees are shown in Fig.~\ref{fig:Binarytree9leaves}.
\begin{figure}[h!]
        \centering
          \includegraphics[width=.8\textwidth]{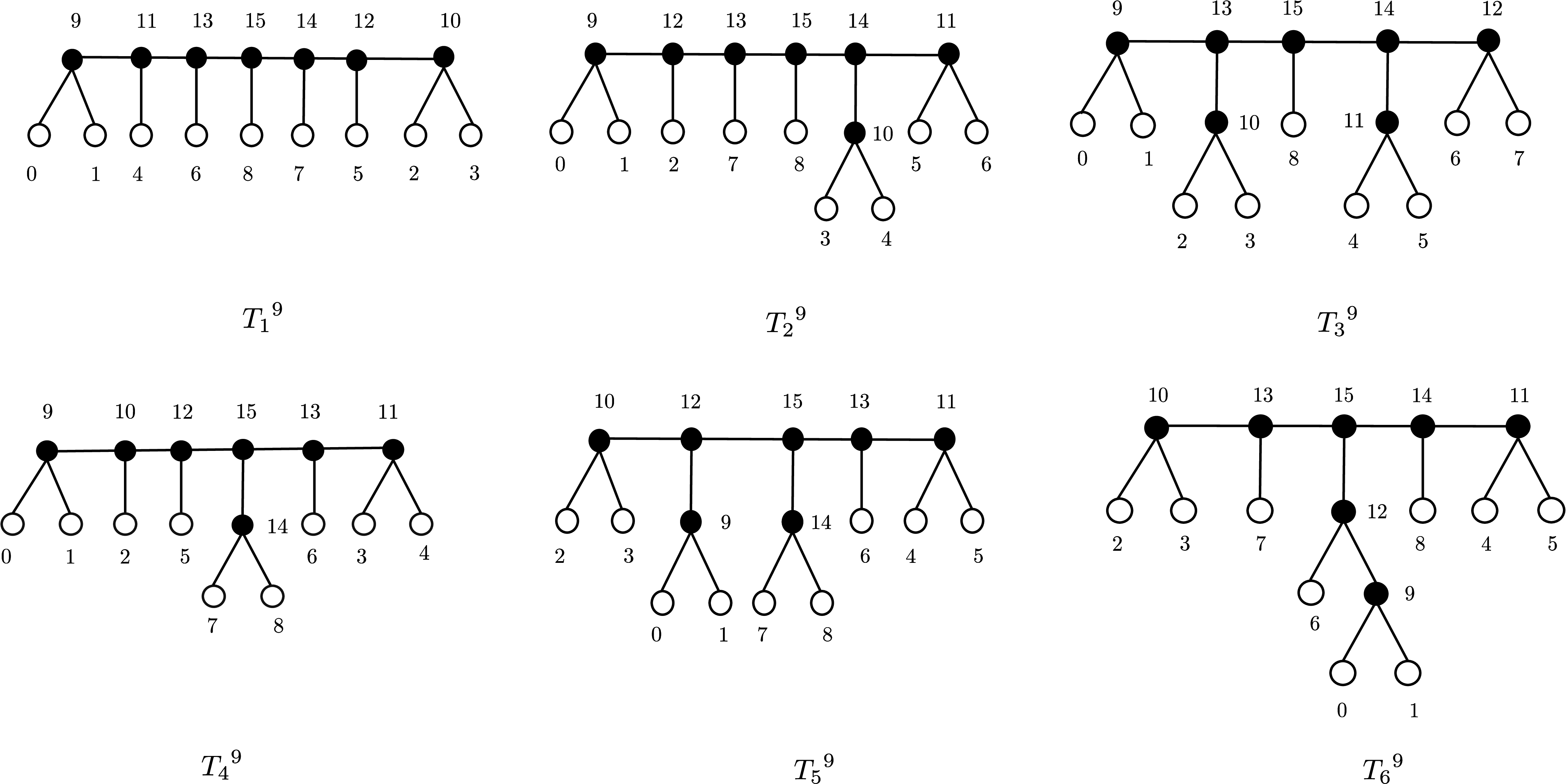}
        \caption{Binary trees with nine leaves}
        \label{fig:Binarytree9leaves}
    \end{figure}
    
    The ranges of the weights are the same as used for the eight vertices. 
    The proposed algorithm identified all  graphs with nine vertices as 2-IPCGs 
    in 54,000 seconds and 36,000 seconds. 
 Table~\ref{table:computation_time_9_vertices1} shows a noticeable effect of weights
 on the number of rounds and the computation time to identify all 2-IPCGs. 
 More precisely, the number of rounds and the computation time with weights 
  $[a_{1}, a_{2}]=[1, 20]$ and  $[a_{3}, a_{4}]=[1, 50]$ are significantly 
  lower than those of the weights 
 $[a_{1}, a_{2}]=[1, 50]$ and $[a_{3}, a_{4}]=[1, 100]$. 
 From Table~\ref{table:time_15_hours}, 
$T_{4}^{9}$ generated the largest number of 2-IPCGs for weights in $[1, 20], [1, 50]$ and $T_{1}^{9}$ generated the largest 2-IPCGs for weights in $[1, 50], [1, 100]$.
These computational results are summarized in Theorem~\ref{the:9ver}.
\begin{theorem}\label{the:9ver}
All graphs with nine vertices are $2$-IPCGs. 
\end{theorem}
As a computational proof of Theorem~\ref{the:9ver}, we provide the witness trees and weights due to which a graph is a 2-IPCG  at  \href{https://drive.google.com/drive/folders/196NEvcqyZXB7v2jOTqY1f11qCMPMT508?usp=drive_link}{evidence}.

\begin{table}[ht]
\centering
\begin{tabular}{|c|c|c|c|c|}
\hline
\multicolumn{5}{|c|}{$|\mathcal{G}|$ = 274,666} \\
\multicolumn{5}{|c|}{Weights $[a_1, a_2]$ = [1, 20] and $[a_3, a_4]$ = [1, 50]} \\
\hline
Sr.no. & Time (sec.) & \#Rounds & \#2-IPCGs & \#Graphs Left \\
\hline
1 & 7,200 & 26 & 269,130 & 5,536 \\
2 & 14,400 & 62 & 274,218 & 448 \\
3 & 21,600 & 102 & 274,573  & 93 \\
4 & 28,800 & 151 & 274,642 & 24 \\
5 & 36,000 & 251 & 274,659 & 7 \\
6 & 50,400 & 740 & 274,665 & 1 \\
7 & 54,000 & 827 & 274,666 & 0 \\
\hline
\multicolumn{5}{|c|}{Weights $[a_1, a_2]$ = [1, 50] and $[a_3, a_4]$ = [1, 100]} \\
\hline
1 & 7,200 & 17 & 268,030 & 6,636 \\
2 & 14,400 & 45 & 274,413 & 253 \\
3 & 21,600 & 75 & 274,627 & 39 \\
4 & 28,800 & 106 & 274,652 & 14 \\
5 & 36,000 & 211 & 274,666 & 0 \\
\hline
\end{tabular}
\caption{Computational results of 2-IPCGs with nine vertices}
\label{table:computation_time_9_vertices1}
\end{table}
\begin{table}[ht]
\centering
\begin{tabular}{|c|c|c|}
\hline
 \multicolumn{2}{|l|}{Weights  [1, 20], [1, 50] }& [1, 50], [1, 100] \\\hline
Tree & \multicolumn{2}{|c|}{\#2-IPCGs} \\
\hline
${T_{1}}^9$ & 53,477 & 63,801 \\
${T_{2}}^9$ & 50,768 & 56,535 \\
${T_{3}}^9$ & 43,683 & 40,639 \\
${T_{4}}^9$ & 57,250 & 47,621 \\
${T_{5}}^9$ & 30,375 & 32,249 \\
${T_{6}}^9$ & 39,113 & 33,821 \\
\hline
\end{tabular}
\caption{Number of 2-IPCGs generated by different trees}
\label{table:time_15_hours}
\end{table}
\pagebreak
\subsection{Experimental results for with ten vertices}
There are 12,005,168 graphs with ten vertices, and  
11 binary trees with ten leaves generated by our proposed algorithm as shown in Fig.~\ref{fig:Binarytree10leaves}. 
\begin{figure}[h!]
        \centering
          \includegraphics[width=.8\textwidth]{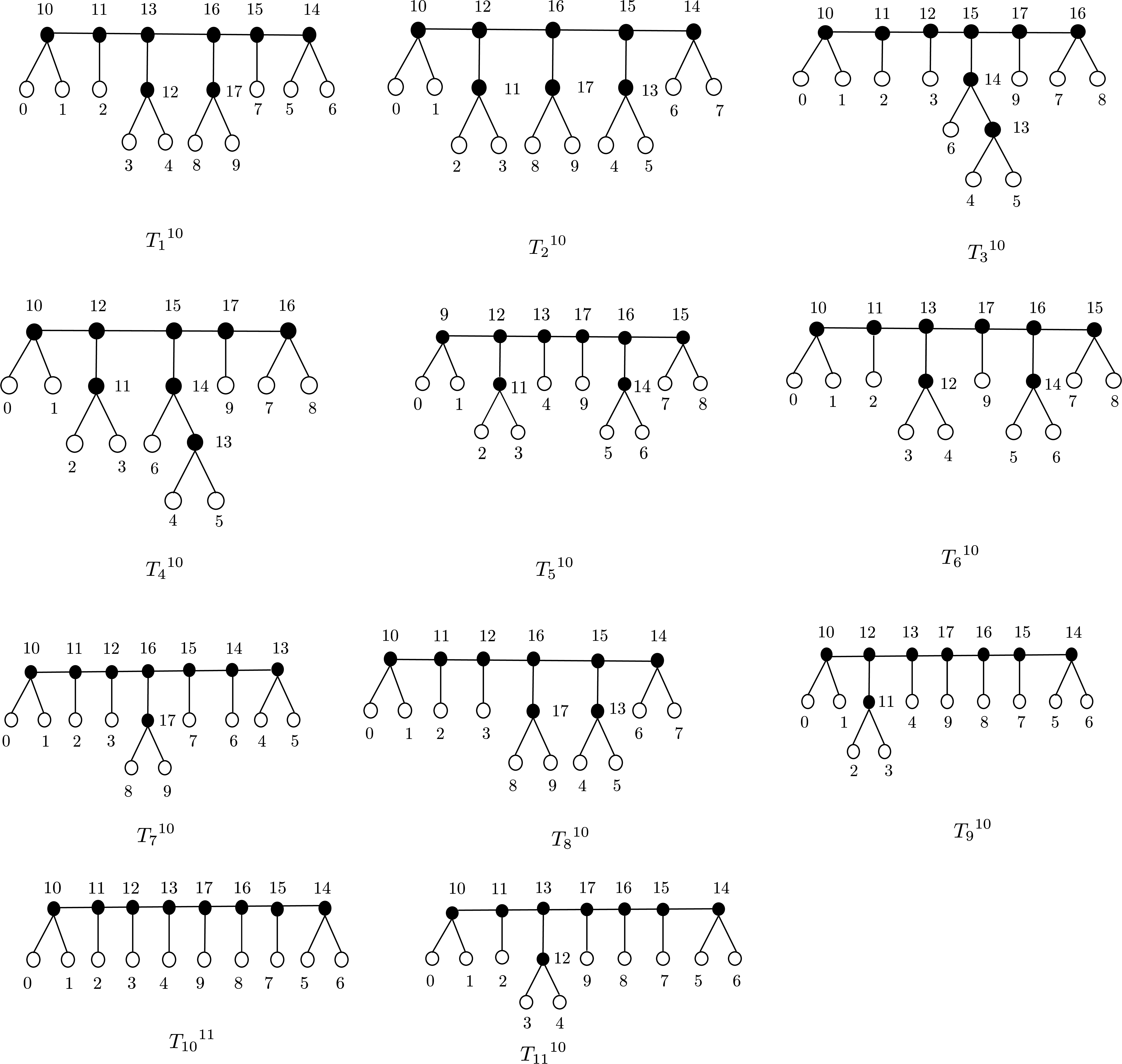}
        \caption{Binary trees with ten leaves}
        \label{fig:Binarytree10leaves}
\end{figure}

It is hard to deal with a large number of graphs on an ordinary machine, therefore we divided the graphs into eight sets, each containing 1,500,000 graphs, except one set which contains 1,505,168 graphs.
Similarly, different combinations of trees with different weights are used to efficiently generate 2-IPCGs. 
The details of these experiments are given in the appendix. 
The number of 2-IPCGs generated by each tree is listed in Table~\ref{table:10_vertices}, 
where $T_4^{10}$ generated the maximum number of 2-IPCGs. 
As a result, all graphs are identified as 2-IPCGs by the proposed algorithm in 96 days. 
These computational results are summarized in Theorem~\ref{the:10ver}.
\begin{table}[ht]
\centering
\begin{tabular}{|c|c|c|c|}
\hline
Tree & \#2-IPCGs & Tree & \#2-IPCGs \\
\hline
${T_{1}}^{10}$ & 2,111,825 & ${T_{7}}^{10}$ & 572,779 \\
${T_{2}}^{10}$ & 2,032,907 & ${T_{8}}^{10}$ & 545,852\\
${T_{3}}^{10}$ & 2,451,760 & ${T_{9}}^{10}$ & 350,927\\
${T_{4}}^{10}$ & 2,152,536 & ${T_{10}}^{10}$ & 345,728 \\
${T_{5}}^{10}$ & 575,431 & ${T_{11}}^{10}$ & 308,483\\
${T_{6}}^{10}$ & 556,938 & &\\
\hline
\end{tabular}
\caption{Number of 2-IPCGs generated by different trees}
\label{table:10_vertices}
\end{table}
\begin{theorem}\label{the:10ver}
All graphs with ten vertices are $2$-IPCGs. 
\end{theorem}
As a computational proof of Theorem~\ref{the:8ver}, we provide the witness trees and weights due to which a graph is a 2-IPCG  at
 \href{https://drive.google.com/drive/folders/11Kdyn3LJKLw9K1rGmzh5T6fqWWUjMDu_?usp=drive_link}{evidence}.
 \pagebreak
\section{Conclusion} \label{sec:conclusion}
We proposed a method to generate $k$-IPCGs. The key idea of this method is to generate almost all $k$-IPCGs for a given tree by randomly assigning weights and selecting intervals based on the distances between leaves of the tree. To reduce the exponential tree search space, we proved that each $k$-IPCG has a binary tree as a witness tree.
This method is applied to graphs with eight, nine, and ten vertices to generate $2$-IPCGs. Through computational experiments, we showed that all graphs with eight, nine, and ten vertices are 2-IPCGs. The method successfully identified all $2$-IPCGs with eight and nine vertices in reasonable times of 700 seconds and 36,000 seconds, respectively.
However, for graphs with ten vertices, the computational time increases significantly to 96 days. 
There are 1,018,997,864 graphs with eleven vertices, making it necessary to improve the method to handle larger numbers of vertices efficiently.
The experiments also reveal that the selection of random weights can affect the computation time of the method. Thus, it would be interesting to further explore the relationship between weight selection and the performance of the method for potential improvements.

\newpage
\centerline{\bf\LARGE Appendix}
Details of the experimental results for graphs with ten vertices. 
\begin{longtable}{|c|c|c|c|c|}
\hline
\multicolumn{5}{|c|}{\textbf{Time: 30 hours}} \\
\hline
\multicolumn{5}{|c|}{$|\mathcal{G}|$ = 12,005,168, Weights [1, 20] and [1, 50]} \\
\hline 
Tree & \#Graphs & \#Rounds & \#2-IPCGs & \#Graphs Left \\
\endfirsthead
\hline
Tree & \#Graphs & \#Rounds & \#2-IPCG & \#Graphs Left \\ \hline
\endhead 
\hline
\multicolumn{5}{|c|}{Continued} \\ \hline
\endfoot
\endlastfoot
\hline
${T_{1}}^{10}$ to ${T_{4}}^{10}$ & 1,500,000 & 317 & 1,439,334 & 60,666 \\
 & 1,500,000 & 167 & 1,117,992 & 382,008 \\
 & 1,500,000 & 264 & 1,068,231 & 431,769 \\
 & 1,500,000 & 192 & 980,412 & 519,588 \\
 & 1,500,000 & 222 & 893,981 & 606,019 \\
 & 1,500,000 & 117 & 751,320 & 748,680 \\
 & 1,500,000 & 258 & 899,023 & 600,977 \\
 & 1,505,168 & 224 & 1,021,259 & 483,909 \\

${T_{5}}^{10}$ to ${T_{8}}^{10}$ & 1,500,000 & 185 & 775,950 & 724,050 \\
 & 1,500,000 & 148 & 584,606 & 915,394 \\
 & 833,616 & 429 & 564,865 & 268,751 \\
${T_{9}}^{11}$ to ${T_{11}}^{10}$ & 1,000,000 & 330 & 744,834 & 255,166 \\
 & 908,195 & 343 & 459,042 & 449,153 \\
\hline
\multicolumn{5}{|c|}{$|\mathcal{G}|$ = 1,000,802, Weights [1, 40] and [1, 70]} \\
\hline 
${T_{1}}^{10}$ to ${T_{4}}^{10}$ & 500,000 & 156 & 187,365 & 312,635 \\
 & 500,804 & 483 & 345,019 & 155,785 \\
${T_{5}}^{10}$ to ${T_{8}}^{10}$ & 253,342 & 436 & 161,756 & 91,586 \\
 & 215,078 & 544 & 143,617 & 71,461 \\
${T_{9}}^{11}$ to ${T_{11}}^{10}$ & 91,586 & 534 & 46,966 & 44,620 \\
 & 71,461 & 741 & 40,008 & 31,453 \\
\hline
\multicolumn{5}{|c|}{$|\mathcal{G}|$ = 76,071, Weights [1, 70] and [1, 100]} \\
\hline 
${T_{1}}^{10}$ to ${T_{4}}^{10}$ & 44,620 & 290 & 20,121 & 24,499 \\
 & 31,453 & 415 & 16,347 & 15,106 \\
${T_{5}}^{10}$ to ${T_{8}}^{10}$ & 24,499 & 282 & 9,389 & 15,110 \\
 & 15,106 & 399 & 6,623 & 8,483 \\
${T_{9}}^{11}$ to ${T_{11}}^{10}$ & 15,110 & 428 & 5,220 & 9,890 \\
 & 8,483 & 533 & 2,892 & 5,591 \\
\hline
\multicolumn{5}{|c|}{$|\mathcal{G}|$ = 15,479, Weights [1, 90] and [1, 120]} \\
\hline 
${T_{1}}^{10}$ to ${T_{4}}^{10}$ & 9,890 & 302 & 3,853 & 6,037 \\
 & 5,591 & 379 & 2,325 & 3,266 \\
${T_{5}}^{10}$ to ${T_{8}}^{10}$ & 6,037 & 298 & 1,896 & 4,141 \\
 & 3,266 & 372 & 1,053 & 2,213 \\
${T_{9}}^{11}$ to ${T_{11}}^{10}$ & 4,141 & 464 & 1,257 & 2,884 \\
 & 2,213 & 372 & 559 & 1,654 \\
\hline
\multicolumn{5}{|c|}{$|\mathcal{G}|$ = 4,536, Weights [1, 110] and [1, 140]} \\
\hline 
${T_{1}}^{10}$ to ${T_{4}}^{10}$ & 2,884 & 297 & 881 & 2,003 \\
 & 1,654 & 353 & 548 & 1,106 \\
${T_{5}}^{10}$ to ${T_{8}}^{10}$ & 2,003 & 281 & 502 & 1,501 \\
 & 1,106 & 337 & 282 & 824 \\
${T_{9}}^{10}$ to ${T_{11}}^{10}$ & 1499 & 436 & 372 & 1,129 \\
 & 824 & 461 & 170 & 654 \\
\hline
\multicolumn{5}{|c|}{$|\mathcal{G}|$ = 1,781, Weights [1, 130] and [1, 160]} \\
\hline
${T_{1}}^{10}$ to ${T_{4}}^{10}$ & 1,127 & 267 & 291 & 836 \\
 & 654 & 323 & 158 & 496 \\
${T_{5}}^{10}$ to ${T_{8}}^{10}$ & 836 & 275 & 164  & 672 \\
 & 496 & 318  & 82  & 414  \\
${T_{9}}^{10}$ to ${T_{11}}^{10}$ & 672 & 446  & 138  & 534  \\
& 414 & 451  & 56  &  358 \\
\hline
\multicolumn{5}{|c|}{$|\mathcal{G}|$ = 892, Weights [1, 150] and [1, 180]} \\
\hline
${T_{1}}^{10}$ to ${T_{4}}^{10}$ & 892 & 118 & 82 & 810 \\
${T_{5}}^{10}$ to ${T_{8}}^{10}$ & 810 & 131 & 75  & 735 \\
${T_{9}}^{10}$ to ${T_{11}}^{10}$ & 735 & 190 & 61  & 674  \\
\hline
\multicolumn{5}{|c|}{$|\mathcal{G}|$ = 674, Weights [1, 170] and [1, 200]} \\
\hline
${T_{1}}^{10}$ to ${T_{4}}^{10}$ & 674 & 130 & 81 & 593 \\
${T_{5}}^{10}$ to ${T_{8}}^{10}$ & 593 & 126 & 56  & 537 \\
${T_{9}}^{10}$ to ${T_{11}}^{10}$ & 537 & 179 & 36  & 501  \\
\hline
\multicolumn{5}{|c|}{$|\mathcal{G}|$ = 501, Weights [1, 190] and [1, 220]} \\
\hline
${T_{1}}^{10}$ to ${T_{4}}^{10}$ & 501 & 129 & 51 & 450 \\
${T_{5}}^{10}$ to ${T_{8}}^{10}$ & 450 & 128 & 36  & 414 \\
${T_{9}}^{10}$ to ${T_{11}}^{10}$ & 414 & 169 & 12  & 402  \\
\hline
\multicolumn{5}{|c|}{$|\mathcal{G}|$ = 402, Weights [1, 210] and [1, 240]} \\
\hline 
${T_{4}}^{10}$ & 402 & 872 & 67  & 335 \\
${T_{3}}^{10}$  & 335 & 833 & 41  & 294 \\
${T_{8}}^{10}$  & 294 & 843 & 31  & 263 \\
${T_{2}}^{10}$  & 263 & 844 & 26 & 237 \\
\hline
\multicolumn{5}{|c|}{ Weights [1, 230] and [1, 260]} \\
\hline 
${T_{1}}^{10}$ & 237 & 831 & 23  & 214 \\
${T_{6}}^{10}$  & 214 & 842 & 15  & 199 \\
${T_{4}}^{10}$  & 199 & 880 & 21  & 178 \\
${T_{3}}^{10}$  & 178 & 919 & 17  & 161 \\
\hline
\multicolumn{5}{|c|}{ Weights [1, 250] and [1, 280]} \\
\hline 
${T_{4}}^{10}$ & 161 & 927 & 32  & 129 \\
${T_{3}}^{10}$  & 129 & 908  & 16  & 113 \\
\hline
\multicolumn{5}{|c|}{\textbf{Time 10 hours}} \\
\hline
\multicolumn{5}{|c|}{ Weights [1, 270] and [1, 300]} \\

\hline 
${T_{4}}^{10}$ & 113 & 319 & 6  & 107 \\
\hline
\multicolumn{5}{|c|}{Weights [1, 290] and [1, 320]} \\
\hline 
${T_{4}}^{10}$ & 107 & 301 & 6  & 101 \\
\hline
\multicolumn{5}{|c|}{ Weights [1, 310] and [1, 340]} \\
\hline
${T_{3}}^{10}$ & 101 & 311 & 3  & 98 \\
\hline
\multicolumn{5}{|c|}{Weights [1, 330] and [1, 360]} \\
\hline
${T_{4}}^{10}$ & 98 & 306 & 7  & 91 \\
\hline
\multicolumn{5}{|c|}{Weights [1, 370] and [1, 400]} \\
\hline 
${T_{3}}^{10}$ & 91  & 318 & 4  & 87 \\
\hline
\multicolumn{5}{|c|}{ Weights [1, 390] and [1, 420]} \\
\hline
${T_{4}}^{10}$ & 87 & 288 & 2  & 85 \\
\hline
\multicolumn{5}{|c|}{Weights [1, 410] and [1, 440]} \\
\hline 
${T_{1}}^{10}$ & 85 & 285 & 4  & 81 \\
\hline
\multicolumn{5}{|c|}{ Weights [1, 450] and [1, 480]} \\
\hline 
${T_{3}}^{10}$ & 81 & 306 & 6  & 75 \\
\hline
\multicolumn{5}{|c|}{\textbf{Time: 5 hours}} \\
\hline
\multicolumn{5}{|c|}{$|\mathcal{G}|$ = 75, Weights [5, 25] and [5, 50]} \\
\hline
${T_{4}}^{10}$ & 75 & 402 & 5 & 70 \\
\hline
\multicolumn{5}{|c|}{Weights [12, 50] and [15, 65]} \\
\hline 
${T_{4}}^{10}$ & 70 & 311 & 3  & 67\\
\hline
\multicolumn{5}{|c|}{ Weights [5, 40] and [20, 70]} \\
\hline
${T_{3}}^{10}$ & 67 & 314 & 3  & 64\\
\hline
\multicolumn{5}{|c|}{ Weights [18, 55] and [25, 80]} \\
\hline
${T_{1}}^{10}$ & 64 & 266 & 2  & 62\\
\hline
\multicolumn{5}{|c|}{ Weights [25, 60] and [10, 85]} \\
\hline
${T_{8}}^{10}$ & 62 & 238 & 2  & 60\\
\hline
\multicolumn{5}{|c|}{ Weights [1, 70] and [20, 100]} \\
\hline
${T_{4}}^{10}$ & 60 & 205 & 5  & 55\\
\hline
\multicolumn{5}{|c|}{ Weights [30, 80] and [1, 115]} \\
\hline
${T_{3}}^{10}$ & 55 & 218 & 3  & 52\\
\hline
\multicolumn{5}{|c|}{ Weights [50, 100] and [1, 130]} \\
\hline
${T_{4}}^{10}$ & 52 & 221 & 2  & 50\\
\hline
\multicolumn{5}{|c|}{ Weights [65, 150] and [55, 180]} \\
\hline
${T_{4}}^{10}$ & 50 & 206 & 2  & 48\\
\hline
\multicolumn{5}{|c|}{ Weights [90, 190] and [100, 215]} \\
\hline
${T_{4}}^{10}$ & 48 & 177 & 1  & 47\\
\hline
\multicolumn{5}{|c|}{ Weights [1, 470] and [1, 500]} \\
\hline
${T_{4}}^{10}$ & 47 & 167 & 1  & 46\\
\hline
\multicolumn{5}{|c|}{ Weights [1, 500] and [1, 530]} \\
\hline
${T_{2}}^{10}$ & 46 & 174 & 1  & 45\\
\hline
\multicolumn{5}{|c|}{ Weights [1, 530] and [1, 570]} \\
\hline
${T_{3}}^{10}$ & 45 & 123 & 1  & 44\\
\hline
\multicolumn{5}{|c|}{ Weights [1, 550] and [1, 590]} \\
\hline
${T_{4}}^{10}$ & 44 & 144 & 1  & 43\\
\hline
\multicolumn{5}{|c|}{ Weights [1, 450] and [1, 600]} \\
\hline
${T_{3}}^{10}$ & 43 & 153 & 1  & 42\\
\hline
\multicolumn{5}{|c|}{ Weights [1, 470] and [1,620]} \\
\hline
${T_{4}}^{10}$ & 42 & 150 & 1  & 41\\
\hline
\multicolumn{5}{|c|}{ Weights [1, 490] and [1, 640]} \\
\hline
${T_{4}}^{10}$ & 41 & 148 & 1  & 40\\
\hline
\multicolumn{5}{|c|}{ Weights [1, 600] and [1, 700]} \\
\hline
${T_{4}}^{10}$ & 40 & 170 & 3  & 37\\
\hline
\multicolumn{5}{|c|}{ Weights [1, 650] and [1, 750]} \\
\hline
${T_{4}}^{10}$ & 37 & 156 & 1  & 36\\
\hline
\multicolumn{5}{|c|}{\textbf{Time: 10 hours}} \\
\hline
\multicolumn{5}{|c|}{$|\mathcal{G}|$ = 36, Weights [1, 1,000] and [1, 1,100]} \\
\hline
${T_{1}}^{10}$ & 36 & 260 & 3 & 33 \\
\hline
\multicolumn{5}{|c|}{Weights [1, 3,000] and [1, 5,000]} \\
\hline 
${T_{1}}^{10}$ & 33 & 287 & 3  & 30\\
\hline
\multicolumn{5}{|c|}{ Weights [1, 20,000] and [1, 25,000]} \\
\hline
${T_{4}}^{10}$ & 30 & 212 & 4  & 26\\
\hline
\multicolumn{5}{|c|}{ Weights [1, 25,000] and [1, 30,000]} \\
\hline
${T_{4}}^{10}$ & 26 & 333 & 3  & 23\\
\hline
\multicolumn{5}{|c|}{ Weights [1, 30,000] and [1, 35,000]} \\
\hline
${T_{1}}^{10}$ & 23 & 346 & 2  & 21\\
\hline
\multicolumn{5}{|c|}{\textbf{Time: 20 hours}} \\
\hline
\multicolumn{5}{|c|}{$|\mathcal{G}|$ = 21, Weights [5, 25] and [5, 50]} \\
\hline
${T_{4}}^{10}$ & 21 & 1,848 & 4  & 17\\
\hline
\multicolumn{5}{|c|}{Weights [1, 70] and [20, 100]} \\
\hline 
${T_{4}}^{10}$ & 17 & 1,318 & 2  & 15\\
\hline
\multicolumn{5}{|c|}{ Weights [30, 80] and [1, 115]} \\
\hline
${T_{3}}^{10}$ & 15 & 1,104 & 1  & 14\\
\hline
\multicolumn{5}{|c|}{ Weights [5, 25] and [5, 50]} \\
\hline
${T_{3}}^{10}$ & 14 & 2,444 & 1  & 13\\
\hline
\multicolumn{5}{|c|}{ Weights [12, 50] and [15, 65]} \\
\hline
${T_{4}}^{10}$ & 13 & 2,015 & 1  & 12\\
\hline
\multicolumn{5}{|c|}{ Weights [5, 25] and [5, 50]} \\
\hline
${T_{1}}^{10}$ & 12 & 2,790 & 4  & 8\\
\hline
\multicolumn{5}{|c|}{ Weights [65, 150] and [55, 180]} \\
\hline
${T_{4}}^{10}$ & 8 & 1,387 & 1  & 7\\
\hline
\multicolumn{5}{|c|}{ Weights [20, 90] and [30, 70]} \\
\hline
${T_{3}}^{10}$ & 7 & 2,195  & 2  & 5\\
\hline
\multicolumn{5}{|c|}{ Weights [1, 1,000] and [1, 1,100]} \\
\hline
${T_{4}}^{10}$ & 5 & 1,189  & 1  & 4\\
\hline
\multicolumn{5}{|c|}{ Weights  [1, 600] and [1, 700]} \\
\hline
${T_{1}}^{10}$ & 4 & 2,652  & 2  & 2\\
\hline
\multicolumn{5}{|c|}{ Weights  [1, 450] and [1, 600]} \\
\hline
${T_{1}}^{10}$ & 2 & 4,559  & 1  & 1\\
\hline
\multicolumn{5}{|c|}{\textbf{Time: 6 mintues}} \\
\hline
\multicolumn{5}{|c|}{$|\mathcal{G}_{10}|$ = 1, Weights [1, 350] and [1, 450]} \\
\hline
${T_{4}}^{10}$ & 1 & 46 & 1  & 0\\
\hline
\caption{Computation Results of 2-IPCGs with ten Vertices}
\label{table:computation_time_10_vertices1}\\
\end{longtable}

\end{document}